%% file: main.tex
\documentclass[letterpaper,12pt,leqno]{article}
\usepackage{paper}
\usepackage{makecell}
\setcitestyle{numbers,square}
\newcommand{\keywords}[1]{\par\vspace{0.5em}\noindent\textbf{Keywords:} #1\par\vspace{0.5em}}
\usepackage{booktabs, multirow}
\usepackage{siunitx}
\usepackage[colorinlistoftodos]{todonotes}
\usepackage{siunitx}
\usepackage{amsmath}      % For math environments
\usepackage{algorithm}    % For the floating algorithm environment
\usepackage{algpseudocode}% For pseudocode syntax

\makeatletter
\algnewcommand{\algorithmicforeach}{\textbf{for each}}
\algnewcommand{\algorithmicendforeach}{\textbf{end for each}}

\algdef{SE}[FOR]{ForEach}{EndForEach}[1]{%
  \algorithmicforeach\ #1 \ \algorithmicdo%
}{%
  \algorithmicendforeach%
}
\makeatother

\DeclareMathOperator*{\argmin}{arg\,min}

\newtheorem{thm}{Theorem}[section]

 \newtheorem{cor}{Corollary}
  
 \usepackage{multirow}
\usepackage{siunitx}
\theoremstyle{definition}

  \newtheorem{rem}{Remark}[section]
  \newtheorem{ass}{Assumption}
  \usepackage{booktabs}
\usepackage{siunitx}

\usepackage{siunitx}
\newcommand{\bib}{refs.bib}

\begin{document}

\include{guido_macros.tex}

\title{Triply Robust  Panel Estimators}
\author{Susan Athey, Guido Imbens, Zhaonan Qu \& Davide Viviano
\thanks{Susan Athey, Stanford University; Guido Imbens, Stanford University; Zhaonan Qu, Columbia University; Davide Viviano, Harvard University.
%many colleagues for helpful comments and discussions. 
We thank
the Office of Naval Research for support under grant numbers N00014-17-1-2131 and N00014-19-1-2468 and Amazon for a gift. An early version of this paper was presented as the Journal of Applied Econometrics lecture at the ASSA meetings in San Francisco in January 2025.
We are grateful for comments by Isaiah Andrews, Dmitry Arkhangelesky, Jesse Shapiro and Aico van Vuren.  
}}
\date{\today}                       
\begin{titlepage} \maketitle This paper studies estimation of causal effects in a panel data setting. We introduce a new estimator, the Triply RObust Panel (TROP) estimator, that combines $(i)$ a flexible model for the potential outcomes based on a low-rank factor structure on top of a two-way-fixed effect specification, with $(ii)$ unit weights intended to upweight units  similar to the treated units  and $(iii)$ time weights intended to upweight time periods close to the treated time periods. We study the performance of the estimator in a set of simulations designed to closely match several commonly studied real data sets.  We find that there is substantial variation in the performance of the estimators across the settings considered. The proposed estimator outperforms Two-Way-Fixed-Effect (TWFE) or Difference-In-Differences (DID), Synthetic Control (SC), Matrix Completion (MC) and Synthetic-Difference-In-Differences (SDID) estimators. We investigate what features of the data generating process lead to this superior performance, and assess the relative importance of the three components of the proposed estimator.
We have two recommendations. Our preferred strategy is   that researchers use simulations closely matched to the data they are interested in, along the lines discussed in this paper, to investigate which estimators work well in their particular setting.
A simpler approach is to use more robust estimators such as SDID or the new TROP estimator which we find to substantially outperform TWFE/DID estimators in many empirically relevant settings.
%\listoftodos

\keywords{Causal Inference, Panel Data, Factor Models, Difference in Differences, Synthetic Controls}

\end{titlepage}
%\pagestyle{plain}

%\section{Introduction}\label{s:introduction}

%\guido{guido comments}\zhaonan{zhaonan comments}\davide{davide comments}\susan{susan comments}

\section{Introduction}

In recent years, there has been a resurgence of interest in methods for analyzing panel data. The focus in this literature has been on estimating causal effects, often of binary interventions. One important part of this literature has been
the introduction of Synthetic Control (SC) methods 
\citep{abadie2003, abadie2010synthetic, abadie2014}, with subsequent work proposing modifications of the basic SC approach. Another recent strand of the panel data literature has focused on modifications of  the older Difference-in-Differences (DID) and Two-Way-Fixed-Effect (TWFE) methods \citep{callaway2020difference, goodman2021difference, sun2020estimating, roth2023s, chen2025efficient, baker2025difference, sant2020doubly, abraham2018estimating, de2019two}.

The different panel data estimators are  motivated by different assumptions that are difficult to validate or even compare in practice.  DID  methods invoke parallel trends, possibly controlling for time heterogeneity; Matrix Completion (MC) and interactive fixed effects estimators rely on a low-rank factor model; SC emphasizes alignment on pretreatment outcomes. As a result empirical researchers  face a difficult choice of which estimator to use when the validity of these assumptions is uncertain.

A second, related, challenge concerns how existing methods weight observations across units and over time.  Classical SC assigns a single set of unit weights that must suit every post-treatment period and implicitly treats all pre-periods as equally informative, while DID imposes uniform weights for each time period and unit. Both thus ignore the fact that outcomes in different periods may have different predictive power for the target counterfactual. 

A third concern is that SC methods and modifications thereof such as Synthetic Difference-In-Differences (SDID), by relying on balancing weights, do not naturally extend to general assignment patterns with units moving into and out of treatment.

Motivated by these observations, we introduce the Triply RObust Panel (TROP) estimator, a class of methods that is flexible in choosing counterfactual comparison groups. The core idea of the TROP estimator is to incorporate three components: time weights, unit weights, and a flexible outcome model. Each of these three components depends on a data-driven tuning parameter.

Specifically, 
TROP predicts outcomes for treated unit–time pairs using a parsimonious regression model ({\it e.g.,} regularized low-rank factor structure) in combination with unit and time fixed effects, while in addition jointly learning two sets of data-driven weights. These weights can, in principle, take many forms and encompass existing estimators as special cases. In practice, we choose the first set of unit weights to minimize imbalance in pretreatment trajectories between treated and control units, and the second set of time weights to discount less informative periods to allow for the possibility that more recent observations carry greater influence. However, other weights, such as time weights that include seasonal components, are also possible. Both weighting schemes depend on tuning parameters chosen via cross-validation and can accommodate complex assignments. 

Together with the regression adjustment that may capture latent common factors, these elements give TROP its defining property of triple robustness: the estimator is asymptotically unbiased if either time weights, unit weights, or the regression adjustment remove underlying biases. Specifically, we show in Theorem \ref{thm:1} that the bias of the estimator depends on the \textit{product} of unit-level imbalance between treated and control units, time-level imbalance between treatment and control periods, and bias arising from a (possibly) misspecified regression adjustment. This multiplicative bias bound is tighter than the bias bounds that govern leading alternatives such as TWFE/DID, SC, and SDID, providing TROP with stronger robustness properties.

To set the stage we compare the TROP estimator to the SDID, SC, Difference In Differences (DID), MC and Doudchenko-Imbens-Fernan-Pinto (DIFP)  estimators in a range of realistic settings based on real data. The complete description of the simulation designs is deferred to Section \ref{sec:numerical_illustration}.  
We compare the estimators in terms of their Root-Mean-Squared-Error (RMSE) to predict the \textit{counterfactual}, \emph{i.e.,} the  potential control outcome  in the  treated periods for the treated units.\footnote{The RMSE here takes the square root of the average of the squared unit/period specific errors.} The counterfactual control outcomes, unobserved in a real-world study, are known to us by design and set so that treatment effects are all zero.  

Table \ref{tab:first} reports results for twenty-one settings that are our primary focus in this study. In twenty out of the twenty-one cases the TROP estimator is the best in terms of RMSE, with in the one remaining case the standard SC estimator better. 
The table documents striking amounts of heterogeneity in the performance of existing methods. 
Compared to the TROP estimator, DID can have an RMSE that is 900\% higher, for the Basque GDP data with random assignment. The RMSE of SC can be more than 500\% times bigger, for the same case. MC and SDID do somewhat better, never more than 300\% worse or 93\% worse than the TROP estimator, respectively.
Note however, that, like the SC estimator, the SDID estimator does not extend naturally to general assignment mechanisms, whereas the TROP estimator, like the MC estimator, allows for general assignment patterns. In summary, the TROP estimator is the best one across all designs but in one case out of twenty-one settings, and even in that case  it only underperforms in RMSE by 24\% relative to the best competitor (SC). 
The fact that TROP performs well across environments suggests  that it effectively adapts to which features (time balancing, unit balancing, regression adjustments, or combinations of these) are most relevant to improve the accuracy of the estimator.

Given these results it is natural to ask the following two questions.
First, what  features of the data generating processes in Table \ref{tab:first} matter the most for  the relative  performance? Second, what components of the TROP estimator drive its better performance? 
In what we see as a template for conducting informative semi-synthetic simulations or Monte-Carlo studies, in Section \ref{sec:why_interactive} we systematically explore the drivers of the relative performance of our proposed method, focusing on both the features of the data generating process and the components of the estimator that lead to a better performance. We find that DID is particularly sensitive to the presence of interactive fixed effects, which lead to violations of parallel trend assumptions. Empirically, we find that interactive fixed effects are important for predicting outcomes in many of the chosen applications, spanning examples with both large and small $N$ and $T$. 
We also find that in settings with few treated periods and few treated units other estimators become more competitive with the TROP estimator.
Regarding the three components of the TROP estimator we find that the presence of the time weights, as well as combinations of the time weights and the regression adjustment typically play a key role in significantly reducing RMSE. Unit weights reduce RMSE in some applications, although they play a less prominent role.

\begin{table}[!ht]
\captionsetup{font=scriptsize}
\begin{centering}
\scalebox{0.85}{\begin{tabular}{
  @{} lll
  S[table-format=3.0]      % N  (no decimals)
  S[table-format=2.0]      % T  (no decimals)
  *{6}{S[table-format=1.2]}               % six RMSE columns, 2-dp rounding
@{}}
\toprule
&&& & & \multicolumn{6}{c}{\textbf{Normalized Out-of-sample RMSE}}\\
\cmidrule(l){6-11}
{Data set} & {Outcome} & {Treatment} & {$N$} & {$T$}
& {TROP} & {SDID} & {SC} & {DID} & {MC} & {DIFP}\\
\midrule
CPS     & {log-wage} & {min wage}     & 50  & 40 & \bfseries 1.00 & 1.14 & 1.44 & \textit{1.91} & 1.26 & 1.22 \\
CPS     & {urate} & {min wage}   & 50  & 40 & \bfseries 1.00 & 1.05 & 1.11 & \textit{1.89} & 1.10 & 1.09 \\
CPS     & {hours} & {min wage}        & 50  & 40 & \bfseries 1.00 & 1.19 & 1.22 & \textit{1.25} & 1.11 & 1.20 \\
CPS     & {log-wage} & {gun law}      & 50  & 40 & \bfseries 1.00 & 1.15 & 1.13 & \textit{2.04} & 1.33 & 1.32 \\
CPS     & {log-wage} & {abortion}     & 50  & 40 & \bfseries 1.00 & 1.14 & 1.43 & \textit{1.99} & 1.25 & 1.24 \\
CPS     & {log-wage} & {random}       & 50  & 40 & \bfseries 1.00 & 1.12 & 1.10 & \textit{1.95} & 1.24 & 1.17 \\
PennWT     & {log-GDP} & {democracy}        & 111 & 48 & \bfseries 1.00 & 1.44 & 1.59 & \textit{7.85} & 1.76 & 1.54 \\
PennWT     & {log-GDP} & {education}        & 111 & 48 & \bfseries 1.00 & 1.51 & 2.10 & \textit{6.90} & 1.60 & 1.57 \\
PennWT     & {log-GDP} & {random}           & 111 & 48 & \bfseries 1.00 & 1.32 & 1.51 & \textit{4.31} & 1.41 & 1.50 \\
Germany & {GDP} & {random}           & 17  & 44 & \bfseries 1.00 & 1.46 & 2.82 & \textit{3.58} & 1.56 & 2.46 \\
Germany & {GDP} & {simulated}        & 17  & 44 & \bfseries 1.00 & 1.79 & 5.82 & \textit{7.58} & 1.93 & 4.72 \\
Germany & {GDP} & {treated unit}     & 17  & 44 & \bfseries 1.00 & 1.13 & 1.82 & \textit{5.17} & 1.27 & 1.89 \\
Basque  & {GDP} & {random}           & 18  & 43 & \bfseries 1.00 & 1.02 & 4.55 & \textit{9.11} & 1.70 & 2.47 \\
Basque  & {GDP} & {simulated}        & 18  & 43 & \bfseries 1.00 & 1.55 & 2.11 & \textit{3.05} & 1.30 & 1.66 \\
Basque  & {GDP} & {treated unit}     & 18  & 43 & \bfseries 1.00 & 1.35 & 2.79 & 1.25 & \textit{3.02} & 2.94 \\
Smoking & {packs pc} & {random}      & 39  & 31 & \bfseries 1.00 & 1.22 & 1.48 & \textit{2.16} & 1.14 & 1.45 \\
Smoking & {packs pc} & {simulated}      & 39  & 31 & \bfseries 1.00 & 1.28 & 1.93 & \textit{2.88} & 1.44 & 1.41 \\
Smoking & {packs pc} & {treated unit}& 39  & 31 & \bfseries 1.00 & 1.04 & 1.72 & \textit{2.18} & 1.27 & 1.30 \\
Boatlift& {log-wage} & {random}      & 44  & 19 & \bfseries 1.00 & 1.34 & 1.62 & 1.35 & 1.04 & \textit{1.62} \\
Boatlift& {log-wage} & {simulated}      & 44  & 19 & \bfseries 1.00 & 1.50 & 1.31 & 1.46 & 1.12 & \textit{1.58} \\
Boatlift& {log-wage} & {treated unit}& 44  & 19 & 1.24 & 1.93 & \bfseries 1.00 & 1.35 & 1.17 & \textit{2.07} \\
\bottomrule
\end{tabular}
}
\par\end{centering}
\caption{RMSEs of six estimators for semi-synthetic simulations with data generated using CPS, PWT, Germany reunification, Basque, Smoking and Boatlift data as described in Section \ref{sec:numerical_illustration}. TROP refers to our proposed estimator, SDID to the Synthetic Difference in Differences, SC to Synthetic Control, DID to Difference in Diferences, MC to Matrix Completion, and DIFP to the SC estimator after recentering its mean ({\it i.e.,} with an intercept). RMSE calculations are out-of-sample (the square-root of the average of the squared errors at the unit/period level) and based on 1000 simulation replications. They are normalized by dividing by the smallest RMSE among all estimators. The simulation design is the same as that used in Table 2 in \cite{arkhangelsky2019synthetic} and formally described in Section \ref{sec:numerical_illustration}. Specifically, 
for CPS and PWT dataset, the treatment is simulated by estimating the probability of treatment as in Equation \eqref{eqn:treatment_assignment} below using data from the treatment variable in the data as indicated in the table. For the remaining designs random denotes a random treatment assignment across all units, simulated denotes a simulated treatment with probabilities corresponding to the estimated weights from a synthetic control regression and treated unit corresponds to use the true treated unit to evaluate the performance of the method (where the RMSE is evaluated for the treated unit over a pre-treatment period where therefore treatment effects must equal zero).
The outcome variable is simulated from a rank 4 factor model as in Equation \eqref{eqn:simulated_Y} estimated using the observed outcome from the data and as indicated in the table. Treatment effects are set to be equal to zero. 
$N$ and $T$ denotes the total number of units and periods. For each simulated dataset we consider ten treated periods and ten treated units.} \label{tab:first}
\end{table}

This paper connects to a broad literature on panel data models, including (i) Synthetic Control and matching-type unit weighting \citep{abadie2003, abadie2010synthetic}, (ii) two-way fixed effect estimators \citep[see {\it e.g.},][]{roth2023s, de2023two}; (iii) doubly-robust combinations of the two \citep{arkhangelsky2019synthetic, ben2018augmented, imbens2023identification}; and (iv) interactive fixed-effects and matrix-completion outcome modeling \citep[{\it e.g.}][]{bai2009panel, moon2015linear, bai2013principal, athey2021matrix}.  However, none of this literature has studied triple robust estimators. 

Specifically, SC methods and related works   \citep[see also][]{ben2023estimating, samartsidis2019assessing, brodersen2015inferring, hazlett2018trajectory, chernozhukov2018t, cattaneo2021prediction, ferman2016revisiting, hahn2017synthetic, kellogg2021combining}  do not accommodate (i) complex assignment patterns and (ii) heterogeneous weights over time. Recognizing the complementary roles of SC-style balancing and outcome modeling, Augmented SC \citep{ben2018augmented} and  SDID \citep{arkhangelsky2019synthetic}  achieve double robustness: bias vanishes if either unit imbalance or time imbalance is negligible. However,  SC and  SDID are (i) not applicable to contexts where treatments may switch on and off; (ii)  do not accomodate a more flexible regression adjustment other than the two-way fixed structure. 
\cite{schenk2023time}  also allows for some time weighting.

 Models using interactive fixed effects, which under mild conditions can approximate any outcome model with sufficient factors (that is, a high rank model), can address some of the limitations of TWFE  models. 
 Such approaches can perform well in moderate sample sizes when a low-rank structure (small number of factors) is a very good approximation, but the approaches ignore assignment mechanisms and often require strong conditions about the signal to noise ratio \citep{bai2009panel, moon2017dynamic}. TROP combines these ideas with synthetic control insights: it leverages a regression adjustment such as a low-rank model. As a result, factor misspecification affects at most one of three channels, leaving the estimator consistent provided either unit or time balancing succeeds.

We organize this paper as follows. In Section \ref{sec:setup} we introduce the main setup and estimator focusing on a single treated unit for expositional convenience. Section \ref{sec:numerical_illustration} presents a large collection of numerical studies calibrated to relevant empirical applications.
Section \ref{sec:why_interactive} investigates the drivers of the superiority of the proposed TROP estimator, first in terms of features of the data to which we apply the methods, and second in terms of the features of the proposed estimator.
Section \ref{sec:formal} formalizes the triple robustness properties of the estimator and discuss theoretical properties. Section \ref{sec:extensions} studies the estimator with multiple treated units, covariates and present algorithms for inference. Section \ref{sec:conclusions} concludes.

\section{The Triply RObust Panel (TROP) Estimator} \label{sec:setup}

In this section we introduce the set up and our proposed estimation procedure.

\subsection{Set Up and Estimand}

We observe an $N\times T$ matrix of outcomes ${\bf Y}$, and an $N\times T$ matrix of binary treatment assignments. 
To fix ideas, consider first a single treated unit/period pair so that 
$\sum_{i,t} W_{it}=1.$
The single treated unit-period pair is not essential for the TROP estimator which can deal with general assignment patterns. This differs from other panel data estimators that typically require staggered adoption of the treatment.  However, the single treated unit/period pair simplifies exposition because multiple treated units require to introduce additional notation that involve choosing different weights, one for each unit to be imputed. We will therefore focus on describing the rationale for a single treated unit and defer the multiple units case to Section \ref{sec:inference2} and Algorithm \ref{alg:alg2}. 

We assume there are no spillovers and no dynamic effects. We postulate the existence of a pair of potential outcomes $(Y_{it}(0),Y_{it}(1))$ for each unit and time period with $\tau_{i,t}=Y_{it}(1)-Y_{it}(0)$ the unit/period pair specific treatment effect. We are interested in the average effect for the treated,
\begin{equation} \label{eqn:tau} \tau=\frac{\sum_{i=1}^N\sum_{t=1}^T W_{it}(Y_{it}(1)-Y_{it}(0))}{\sum_{i=1}^N\sum_{t=1}^T W_{it}}.
\end{equation}

\subsection{The TROP Estimator}

The core idea behind our estimator is to introduce a set of time weights across different periods $t$, denoted by $\theta_t$ and different weights across units units $j$, denoted by $\omega_{j}$, and to introduce a regression adjustment. We specify our working model for the potential control outcome as 
$$
Y_{it}(0) = \alpha_i + \beta_t + \mathbf{L}_{it} + \varepsilon_{i,t}, \quad \mathbb{E}[\varepsilon_{it} | \mathbf{L}] = 0
$$ 
where we think of $\varepsilon_{it}$ as idiosyncratic shocks and $\mathbf{L}$ as a component we will try to predict ({\it e.g.}, through a low-rank factor model). Although the unit and time fixed effects $\alpha_i$ and $\beta_t$ can be incorporated into the low rank component ${\bf L}$ (adding a rank two component), because in practice we introduce regularization to recover $\mathbf{L}$ but not the unit and time fixed effects it will be convenient to include them separately unless otherwise specified. In principle, one could consider several estimators for $\mathbf{L}$, and our theory will consider a broad class of estimators. 

In practice (and in our numerical studies), 
we propose estimating the treatment effect on unit $(i,t)$ by minimizing 
\begin{equation}   
(\hat\alpha,\hat\beta,\hat{\bf L}) = \mathrm{arg} \min_{\alpha,\beta,\mathbf{L}}\sum_{j=1}^{N}\sum_{s=1}^{T}  \theta_s^{i, t} \omega_j^{i, t}(1-W_{js}) \left(Y_{js}-\alpha_{j}-\beta_{s}-L_{js}\right)^{2}+\lambda_{nn}\left\Vert \mathbf{L}\right\Vert \label{eq:doubly-weighted-regression},
\end{equation} 
where $\omega_{j}^{i, t}, \theta_s^{i, t}$ allows for heterogeneous weights to different units and periods, and $\lambda_{nn}$ imposes a penalization on the nuclear norm of the matrix $\mathbf{L}$. Given estimates $\hat\alpha_i$, $\hat\beta_t$ and $\hat{\bf L}_{it}$ we estimate the unit/period treatment effect for a treated unit period pair $(i,t)$  as
\[ \hat{\tau}_{it}=Y_{it}-\hat{\alpha}_i-\hat{\beta}_t-\hat {\bf L}_{it}.\]
The weights, can be implicitly functions of $(i, t)$ and can take arbitrary form. We defer a comprehensive discussion on the choice of the weights to Section \ref{sec:choice_weights}.

The formulation in Equation \eqref{eq:doubly-weighted-regression} encompasses
DID, SC, MC, and SDID estimators as special cases. 
That is, for $\lambda_{nn} = \infty$ and $\omega_{j}= \theta_s = 1$, we recover the DID/TWFE estimator. For $\omega_{j}= \theta_s = 1$, and $\lambda_{nn} <\infty$, we recover the MC estimator \cite{athey2017matrix}. For $\lambda_{nn} = \infty$  we recover the SC and SDID estimators for specific choices of the units and time weights.

To gain further insight into the TROP estimator, note that, as we discuss further in the following section, we can interpret the weights $\omega$ as balancing weights, in the spirit of \cite{zubizarreta2015stable, athey2018approximate}, the weights $\theta$ as time-varying weights, and the penalization of the nuclear norm on $\mathbf{L}$ as imposing a low-rank outcome model structure. That is, the estimator combines flexible balancing with the flexible outcome model in the same spirit of doubly robust estimation \cite{bang2005doubly, chernozhukov2018double}.

\subsection{Choice of the Weights} \label{sec:choice_weights} 

One of the challenge of our proposed estimation procedure is that, in principle, it may depend on several tuning parameters. In particular, it may depend on the choice of the weight as well as wel as on the choice of the rank for $\mathbf{L}$. In practice, often due to data feasibility, we recommend to reduce the number of tuning parameters by imposing some particular functional form restrictions on the choice of the weights.

Here, we propose one specific choice of the weights, although our analysis is not specific to this choice of the weights, and, depending on the particular context, one may wish to generalize the time weights to accomodate seasonal patterns.  

Given tuning parameters $\lambda = (\lambda_{time},\lambda_{unit},\lambda_{nn})$,
the distance-based weights exhibit exponential decay for both the time and units weights:
\begin{equation} \label{eqn:weights1}
\theta_{s}^{i,t}(\lambda)  = \exp\Big(-\lambda_{time}\text{dist}^{\text{time}}(s,t)\Big), \quad {\rm and}\ \ \omega_j^{i,t}(\lambda) = \exp\Big(-\lambda_{unit}\text{dist}_{-t}^{\text{unit}}(j,i)\Big).
\end{equation}
For the time distance we take a simple function that directly measures the number of periods:
\[
\text{dist}^{\text{time}}(s,t)  =|t-s|.\]
For the unit weights we use the square root of the average squared difference for the periods where both units are in the control group:
\[ 
\text{dist}_{-t}^{\text{unit}}(j,i) =\left(\frac{\sum_{u=1}^{T} 1\{u \neq t\} (1-W_{i u})(1-W_{ju})(Y_{i u }-Y_{ju})^{2}}{\sum_{u=1}^{T} 1\{u \neq t\} (1-W_{i u})(1-W_{ju})} \right)^{1/2}.
\] 
where, with an abuse of notation, we write the weights as an explicit function of $\lambda$. 

For this particular choice of weights, units that far apart (either over the time dimension or the unit-level dimension) with respect to the treated units should we weighted less in our final estimator.

Our primary task is designing
an efficient method to choose the triplet $(\lambda_{time},\lambda_{unit},\lambda_{nn})$
based on leave-one-out cross validation. The core idea of the leave-one-out cross validation we propose here is that the predicted treatment effect on control units should be close to zero if were to consider control units as if they were treated. Define, given
$I_{j,s}^{i,t} = 1\{(j,s) = (i,t)\}$,
\begin{equation} \label{eqn:auxiliary_tau}
\begin{aligned} 
 \hat{\tau}_{it}^{\text{loocv}}(\lambda) =  \mathrm{arg} \min_\tau  \min_{\alpha,\beta,\mathbf{L}} \sum_{j,s:W_{j,s}=0}  \omega_{j}^{i,t}(\lambda) \theta_s^{i,t}(\lambda) \left(Y_{js}-\alpha_{j}-\beta_{s}-L_{js} - \tau \cdot I_{j,s}^{i,t} \right)^{2}+\lambda_{nn}\left\Vert \mathbf{L}\right\Vert 
\end{aligned} 
\end{equation}
Here $I_{j,s}^{i,t}$ denotes the indicator that unit $(j,s)$ corresponds to $(i,t)$ and 
$\hat{\tau}_{it}^{\text{loocv}}(\lambda)$ corresponds to the estimated treatment effects on a given unit $(i,t)$, \textit{as if} unit $(i,t)$ was the treated unit, after removing the outcomes for the other treated observations.
Intuitively, because we estimate the treatment effect $\hat{\tau}_{it}^{\text{loocv}}(\lambda)$ only over control units, such effect captures the difference between the true outcome and predicted counterfactual, after leaving out from such regression information at time $t$. 

This intuition motivates the choice of  
$(\lambda_{time},\lambda_{unit},\lambda_{nn})$
that minimizes the criterion function 
\begin{equation} \label{eqn:q_lambda}
Q(\lambda)=\sum_{i=1}^{N}\sum_{t=1}^{T}(1-W_{it}) \Big(\hat{\tau}_{it}(\lambda)\Big)^2 
\end{equation} 
over a grid of values. This is equivalent to choosing the tuning parameters with the smallest out-of-sample squared error for predicting the potential outcome under control on the control units.\footnote{In practice, we start with the following procedure: letting $\lambda_{nn}=\infty$
and $\lambda_{unit}=0$, we minimize $Q$ over a grid of values
of $\lambda_{time}$, and similarly for the other two penalty
parameters. These optimized values provide \emph{upper} bounds for
a finer grid for joint search. Next, starting from these initial values, we select the penalty parameters
by cycling through the three parameters. In other words, we successively
update the three parameters holding the other two at the most recent
optimal values.}

\begin{algorithm} [!ht]   \caption{TROP estimator for treatment effect with single treated unit}\label{alg:alg1}
    \begin{algorithmic}[1]

   \Require Grid of values $\mathcal{G}$ for $(\lambda_{time},\lambda_{unit},\lambda_{nn})$, treatments $\mathbf{W}$, outcomes $\mathbf{Y}$   
   \ForEach{$\lambda \in \mathcal{G}$}
   \State For each $(i,t)$ such that $W_{it} = 0$ estimate $\hat{\tau}_{it}(\lambda)$ as in Equation \eqref{eqn:auxiliary_tau}
   \State Compute the function $Q(\lambda)$ in Equation \eqref{eqn:q_lambda}
 \EndFor 
 \State Find $\hat{\lambda} \in \mathrm{arg} \min_{\lambda \in \mathcal{G}} Q(\lambda)$
\State Compute the estimator $\hat{\tau}(\hat{\lambda})$ as in Equation \eqref{eq:doubly-weighted-regression} with $\lambda$ for the weights replaced by $\hat{\lambda}$ found in the previous step 
         \end{algorithmic}
\end{algorithm}

\section{Semi-synthetic Simulations} \label{sec:numerical_illustration}

The goal of this section is to study the performance of our method in realistic settings. To this end, we first consider 
simulation studies identical to those already considered in \cite{arkhangelsky2019synthetic}. These simulations are calibrated to datasets representative of those typically
used for panel data studies.   

\subsection{Simulation design} \label{sec:design}

As a first step, we summarize the simulation designs from \cite{arkhangelsky2019synthetic} used to create placebo studies based on (1)~PWT World Table (PWT) data for international log-GDP, and
(2)~aggregated Current Population Survey (CPS) data for U.S.\ states.\footnote{For CPS data, we use data for female workers in the March outgoing rotation groups of the
Current Population Survey (CPS), spanning selected years in 1979--2019.
The original dataset includes average (log) wages by state-year cells, where $i$ indexes states
and $t$ indexes years.} In both settings, we generate potential
outcomes via a latent factor model and assign treatment via a nonrandom mechanism that
replicates empirically observed policy patterns. As noted in \cite{arkhangelsky2019synthetic}, the first simulation mimics the relevant settings for when DID is used in practice, and the second for when synthetic control is used in practice.

 Both the CPS-based and PWT-based placebo studies yield balanced panel data with: 
(1)~outcomes driven by a latent-factor structure plus autocorrelated errors, and
(2)~treatment indicators that depend non-randomly on either fixed effects or the interactive
factor component. These designs allow for systematic confounding that would
invalidate simpler difference-in-means comparisons,
while still retaining enough structure (through the factor model) to examine estimators that
correct for confounding.

\paragraph{Latent-Factor Outcome Model}
Let $Y_{it}^*$ be the observed outcome (e.g., average log-wage) in state~$i$ and year~$t$. As in \cite{arkhangelsky2019synthetic}, we approximate
$Y_{it}^*$ via a low-rank factor model
\[
  L_{it} \;=\; \arg\min_{L: \,\mathrm{rank}(L) \le 4}
  \sum_{i,t} \bigl(Y_{it}^* - L_{it}\bigr)^2,
\]
thus obtaining a matrix $\mathbf{L} = (L_{it})$ that captures common trends and cross-sectional
variation. We estimate an autoregressive model for the residuals $Y_{it}^* - L_{it}$ to capture residual
serial correlation within each state. Denote the resulting covariance matrix by $\Sigma$. The
simulated outcomes $Y_{it}$ in the placebo experiments then follow
\begin{equation} \label{eqn:simulated_Y}
  Y_{it} \;=\; L_{it} \;+\; \varepsilon_{it},
\end{equation}
with $\Big(\varepsilon_{i,1}, \cdots, \varepsilon_{i,T}\Big) \sim \mathcal{N}(0, \Sigma)$ from a multivariate Normal distribution with the estimated covariance structure $\Sigma$. As in \cite{arkhangelsky2019synthetic}, for interpretation, we further decompose $\mathbf{L}$ into an additive component and interactive component, where 
\begin{equation} \label{eqn:F_and_M}
F_{i,t} := \alpha_i + \beta_t = \frac{1}{T} \sum_{t=1}^T L_{it} + \frac{1}{N} \sum_{i=1}^N L_{i,t} - \frac{1}{NT} \sum_{it} L_{it}, \quad M_{i,t} := L_{i,t} - F_{i,t}.  
\end{equation}

 We combine the factor structure with idiosyncratic noise exhibiting serial correlation (modeled as an AR(2) process) to reflect realistic dynamics. All outcomes are normalized to have mean zero and unit variance for comparability. We consider treatment effects equal to zero.

\paragraph{Treatment Assignment}

We create a (placebo) treatment indicator $D_i$ for each state $i$ by employing a logistic
regression on the latent factors and fixed effects. Specifically,
\begin{equation} \label{eqn:treatment_assignment}
  D_i \;\Big|\;(\alpha_i,\,M_i, \varepsilon_{i, \cdot})
  \;\sim\; \mathrm{Bernoulli}(\pi_i),
  \quad\text{where}\quad
  \pi_i \;=\;
  \frac{\exp\bigl(\phi_\alpha\,\alpha_i \;+\; \phi_M\,M_i\bigr)}{
        1 \;+\; \exp\bigl(\phi_\alpha\,\alpha_i \;+\; \phi_M\,M_i\bigr)},
\end{equation} 
and $\alpha_i$ and $M_i$ respectively capture additive and interactive components of the
fitted factor model. We estimate $(\phi_\alpha,\phi_M)$ by regressing actual state-level
policy indicators on $\alpha_i$ and $M_i$. Once $D_i$ is drawn, we
treat states with $D_i=1$ as ``exposed'' to a placebo intervention for all $t > T_0$.

For the CPS data, the model for the treatment is estimated using existing policies such as minimum wage policy (default), gun law policies and abortion policies  \citep{arkhangelsky2019synthetic}.  For the PWT Data, we estimate the model using democracy status and education metrics also similar to \cite{arkhangelsky2019synthetic}.  
We examine both a realistic assignment settings, as well as randomized assignments.

In our baseline setup with CPS data, we simulate panel data with $N=50$ units and $T=40$ periods (including $T_{\rm pr}=30$ pre-treatment periods and $T_{\text{post}}=10$ post-treatment periods), with at most $N_{\text{tr}}=10$ number of treated units. For the PWT data, we have $N = 111, T = 48$, and keep $T_{\text{post}}=10$ and $N_{\text{tr}}=10$.  In addition, we vary $T_{\text{post}}, N_{\text{tr}}$ as we describe below.

\paragraph{Four additional datasets: Germany, Basque, Smoking and Boatlift datasets} In addition to the CPS and PWT Dataset, we also conduct simulations using the Germany reunification dataset studied by \cite{abadie2015comparative}, the Basque country dataset studied in \cite{abadie2010synthetic}, the smoking dataset in \cite{abadie2010synthetic} and the boatlift dataset \citep{peri2015labor}. 

A key challenge in these applications is that there is a single treated unit. Therefore, it is not possible to estimate the probability of treatment for these designs. Instead, we consider three DGPs for the treatment assignment matrix. First, we consider settings where the treatment is randomly drawn from the pool of control units (``uniform random''). Second, we consider the one where treatment probability equals the weight estimated through a penalized Synthetic Control regression. This emulates settings where treatment may be confounded by some unobservable components. Third, we report results in terms of predicting the outcome of the treated unit over pre-treatment periods (therefore fixing the choice of the treated unit as the one in the given application). We simulate the outcomes following the same design as for the CPS and PWT dataset, calibrated here to the corresponding Germany, Basque, Smoking, and Boatlift datasets.

\subsection{Main simulation results}

Table \ref{tab:first} collect results across all specification considered (both across different outcomes and treatments), twenty-one in total and focusing on ten treated units and ten treatment periods. The TROP estimator outperforms all competitors in all but only one specification, for which it underperforms by 24\%. In contrast, DID can have RMSE up $900\%$ higher than the best competitor, SC 580\%, MC $300\%$, and SDID $90\%$. (In addition, Appendix Table \ref{tab:sim_bias_only} shows that both TROP and SDID have substantially smaller (no) bias compared to the other estimators focusing for brevity on CPS and PWT data. A a standard TWFE estimator shows noticeable bias when the treatment is not randomly assigned.)

We then zoom in into seven specifications, one with CPS data using log-wage and unemployment rate as outcomes, and minimum wage as treatment; PWT data using log-GDP as outcome and democracy as treatment, and the Smoking, Boatlift, Basque and Germany datasets with randomly assigned treatments (keeping in mind that Table \ref{tab:first} and the Appendix show robust results for a wide range of other non-random assignment mechanisms). 
Specifically, 
  Table 
\ref{tab:merged_panels_Tpost_Ntr} zoom-in on six specific designs, with Table \ref{tab:merged_panels_Tpost_Ntr} providing details about the outcome-generating process and reporting RMSE in absolute terms (instead of relative terms), together with cross-validated values of parameters $(\lambda_{unit},\lambda_{time},\lambda_{nn})$ for the TROP estimator. It collects results for both one and ten treated units and treatment periods.
The TROP estimator also with one treated unit and period outperforms all competitors in half of the cases, and it slightly under-performs (by no more than $25\%$) in the remaining cases.
Appendix Table \ref{tab:two_panels_Tpost_Ntr} reports results for one treated periods, and ten treated units and vice-versa, illustrating similar results to those observed with ten treatment periods and ten treated units. 

In addition, across all designs, we find substantial variability in the optimal choice of the tuning parameters. Unit, time, and regression adjustments all may play a role for the estimator depending on the choice of the outcome and dataset, in some but not all cases.  These results are suggestive of the importance of constructing an estimator that learns the optimal components (time, unit weights, and factor model) from the data directly. 
 In the next section we explore the main drivers of the performance.

\begin{table}[H]
\captionsetup{font=scriptsize}
\centering

% -------- Panel A: T_post=10, N_tr=10 --------
\scalebox{0.6}{
\begin{tabular}{|c|c|c|c|c|c|c|c|c|c|c|c|c|c|c|}
\hline
&&&&&&&&&
\multicolumn{6}{|c|}{\Large $T_{\mathrm{post}}=10,\;N_{\mathrm{tr}}=10$}\\
\hline
\multicolumn{1}{|c|}{Dataset} &
\multicolumn{1}{c|}{Treatment} &
\multicolumn{1}{c|}{$N$} &
\multicolumn{1}{c|}{$T$} &
\multicolumn{1}{c|}{$\frac{\|F\|_{F}}{\sqrt{NT}}$} &
\multicolumn{1}{c|}{$\frac{\|M\|_{F}}{\sqrt{NT}}$} &
\multicolumn{1}{c|}{$\sqrt{\frac{Tr(\Sigma)}{T}}$} &
\multicolumn{1}{c|}{AR(2)} &
\multicolumn{1}{c|}{$(\lambda_{unit},\lambda_{time},\lambda_{nn})$} &
\multicolumn{6}{c|}{RMSE} \\
\hline
 &  &  &  &  &  &  &  &  &
\multicolumn{1}{c|}{TROP} &
\multicolumn{1}{c|}{SDID} &
\multicolumn{1}{c|}{SC} &
\multicolumn{1}{c|}{DID} &
\multicolumn{1}{c|}{MC} &
\multicolumn{1}{c|}{DIFP} \\
\hline
CPS logwage & min wage & 50 & 40 & 0.992 & 0.10 & 0.098 & $(0.010,-0.057)$ & $(0,0.1,0.9)$         & \textbf{0.025} & 0.029 & 0.037 & 0.049 & 0.032 & 0.032 \\
\hline
CPS urate   & min wage & 50 & 40 & 0.770 & 0.40 & 0.604 & $(0.039,0.011)$   & $(1.6, 0.35, 0.011)$ & \textbf{0.203} & 0.214 & 0.226 & 0.384 & 0.224 & 0.221 \\
\hline
PWT         & democracy&111 & 48 & 0.972 & 0.229& 0.069 & $(0.913,-0.221)$  & $(0.3,0.4,0.006)$    & \textbf{0.023} & 0.036 & 0.040 & 0.198 & 0.045 & 0.039 \\
\hline
Germany     & random   & 17 & 44 & 0.994 & 0.108& 0.032 & $(0.840,-0.200)$  & $(1.2,0.2,0.011)$    & \textbf{0.025} & 0.037 & 0.065 & 0.073 & 0.032 & 0.038 \\
\hline
Basque      & random   & 18 & 43 & 0.986 & 0.167& 0.037 & $(1.228,-0.577)$  & $(0,0.35,0.006)$     & \textbf{0.041} & 0.070 & 0.109 & 0.145 & 0.069 & 0.082 \\
\hline
Smoking     & random   & 39 & 31 & 0.935 & 0.337& 0.166 & $(0.506,-0.058)$  & $(0.25, 0.4, 0.011)$ & \textbf{0.085} & 0.108 & 0.132 & 0.192 & 0.102 & 0.129 \\
\hline
Boatlift    & random   & 44 & 19 & 0.913 & 0.344& 0.276 & $(0.031,-0.087)$  & $(0.2, 0.2, 0.151)$  & \textbf{0.115} & 0.155 & 0.187 & 0.156 & 0.121 & 0.187 \\
\hline
\end{tabular}
}

\vspace{0.8ex}

% -------- Panel B: T_post=1, N_tr=1 --------
\scalebox{0.6}{
\begin{tabular}{|c|c|c|c|c|c|c|c|c|c|c|c|c|c|c|}
\hline
&&&&&&&&&
\multicolumn{6}{|c|}{\Large $T_{\mathrm{post}}=1,\;N_{\mathrm{tr}}=1$}\\
\hline
\multicolumn{1}{|c|}{Dataset} &
\multicolumn{1}{c|}{Treatment} &
\multicolumn{1}{c|}{$N$} &
\multicolumn{1}{c|}{$T$} &
\multicolumn{1}{c|}{$\frac{\|F\|_{F}}{\sqrt{NT}}$} &
\multicolumn{1}{c|}{$\frac{\|M\|_{F}}{\sqrt{NT}}$} &
\multicolumn{1}{c|}{$\sqrt{\frac{Tr(\Sigma)}{T}}$} &
\multicolumn{1}{c|}{AR(2)} &
\multicolumn{1}{c|}{$(\lambda_{unit},\lambda_{time},\lambda_{nn})$} &
\multicolumn{6}{c|}{RMSE} \\
\hline
 &  &  &  &  &  &  &  &  &
\multicolumn{1}{c|}{TROP} &
\multicolumn{1}{c|}{SDID} &
\multicolumn{1}{c|}{SC} &
\multicolumn{1}{c|}{DID} &
\multicolumn{1}{c|}{MC} &
\multicolumn{1}{c|}{DIFP} \\
\hline
CPS logwage & min wage & 50 & 40 & 0.992 & 0.10 & 0.098 & $(0.010,-0.057)$ & $(0, 0.2, 0.105)$     & \textbf{0.111} & 0.121 & 0.127 & 0.156 & 0.123 & 0.121 \\
\hline
CPS urate   & min wage & 50 & 40 & 0.770 & 0.40 & 0.604 & $(0.039,0.011)$   & $(2, 0.03, 0.325)$  & 0.637          & 0.713 & 0.688 & \textbf{0.629} & 0.701 & 0.666 \\
\hline
PWT         & democracy&111 & 48 & 0.972 & 0.229& 0.069 & $(0.913,-0.221)$  & $(1.8,4.4,0.01)$      & \textbf{0.052} & 0.056 & 0.116 & 0.294 & 0.059 & 0.083 \\
\hline
Germany     & random   & 17 & 44 & 0.994 & 0.108& 0.032 & $(0.840,-0.200)$  & $(0.1,3.8,0.01)$      & \textbf{0.037} & 0.044 & 0.105 & 0.136 & 0.056 & 0.076 \\
\hline
Basque      & random   & 18 & 43 & 0.986 & 0.167& 0.037 & $(1.228,-0.577)$  & $(0.3,0.2,0.006)$     & 0.039          & \textbf{0.033} & 0.180 & 0.333 & 0.055 & 0.072 \\
\hline
Smoking     & random   & 39 & 31 & 0.935 & 0.337& 0.166 & $(0.506,-0.058)$  & $(0.5, 0.5, 0.0163)$  & 0.200          & \textbf{0.190} & 0.279 & 0.464 & 0.208 & 0.218 \\
\hline
Boatlift    & random   & 44 & 19 & 0.913 & 0.344& 0.276 & $(0.031,-0.087)$  & $(0.1, 0.4, 0.2014)$  & \textbf{0.324} & 0.362 & 0.394 & 0.494 & 0.351 & 0.363 \\
\hline
\end{tabular}
}

\caption{RMSEs of six estimators for semi-synthetic simulations with data generated using CPS, PWT, Germany reunification, Basque, Smoking and Boatlift data as described in Section \ref{sec:numerical_illustration}. TROP refers to our proposed estimator, SDID to the Synthetic Difference in Differences, SC to Synthetic Control, DID to Difference in Diferences, MC to Matrix Completion, and DIFP to the SC estimator after recentering its mean ({\it i.e.,} with an intercept). RMSE calculations are out-of-sample (the square-root of the average of the squared errors at the unit/period level) and based on 1000 simulation replications. The simulation design is the same as that used in Table 2 in \cite{arkhangelsky2019synthetic} and formally described in Section \ref{sec:numerical_illustration}. Specifically, the outcome variable is simulated from a rank-4 factor model as in Equation \eqref{eqn:simulated_Y} estimated using the observed outcome in the data (log-wage and unemployment rate for CPS data, log-GDP for PWT data, GDP for Germany and Basque dataset, cigarettes sold for the smoking dataset and log-wage for the Boatlift dataset). For CPS and PWT dataset, the treatment is simulated by estimating the probability of treatment as in Equation \eqref{eqn:treatment_assignment} using data from the treatment variable (minimum wage and democracy). For Germany, Basque, Smoking and Boatlift dataset, treatment is randomly assigned with uniform probabilities across units.
Columns 5 - 8 defines features of the generating process for the outcome estimated as described in Section \ref{sec:design}. $(\lambda_{unit}, \lambda_{time}, \lambda_{nn})$ denote the estimated tuning parameters for TROP, estimated via cross-validation. $N$ and $T$ denotes the total number of units and periods. The first panel corresponds to ten post-treatment periods and ten treated units and the second panel corresponds to one post-treatment period and one treated unit. }
\label{tab:merged_panels_Tpost_Ntr}
\end{table}

 %In Table \ref{tab:sim_large1} we explore simulation studies using the CPS data. In our (and \cite{arkhangelsky2021synthetic}'s) realistic specifications in Table \ref{tab:sim_large1} which include interactive fixed effects, TROP outperforms all competitors across \textit{all} outcomes and treatment assignment mechanisms considered. 

\section{Drivers of the Superior Performance of the TROP Estimator} \label{sec:why_interactive}

Why do the six estimators we compare perform so differently across the  simulation designs, and in particular, why does the TROP estimator perform so well relative to the other estimators? 
Is it features of the data that drive this, and might it be different for other data sets?
Is it features of the estimator that particularly matter, and could one build those features into other estimators?
In this section we investigate these questions by systematically exploring changes in the data generating processes and in the TROP estimator.
The goal is to give the empirical researcher assistance in choosing the estimator in their own research (which may differ substantially from that covered by the six designs explored in the previous section).

\subsection{Critical Features of the Data Generating Process}

First, we inspect the key components in the data-generating process that drive the performance of the estimators. To do so, we first ``shut down'' relevant components in the data-generating process, with the results presented in Table \ref{tab:sim_shutdown} for the CPS and PWT datasets.
Specifically, starting from the baseline specification (for both CPS and PWT data), we first remove correlation in the residuals and observe similar behavior as in our main simulation results. 

We then remove interactive effects  (the $M$ component in Equation \eqref{eqn:F_and_M}). In this case we see that DID (and other methods such as MC) perform competitively with TROP. For DID this behavior is expected as DID or related methods are designed to tackle additive (non-interactive) unobserved confounders. 
For MC this may be tied to the regularization we use for the factor component.
This result highlights the importance of interactive fixed effects in driving the superior performance of our estimator over existing ones.

\begin{table}
\captionsetup{font=scriptsize}
\centering
\scalebox{0.7}{
\begin{tabular}{|l|S[table-format=1.3]|S[table-format=1.3]|S[table-format=1.3]|l|l|S[table-format=1.3]|S[table-format=1.3]|S[table-format=1.3]|S[table-format=1.3]|S[table-format=1.3]|S[table-format=1.3]|}
\hline
\multicolumn{1}{|c|}{Design} &
\multicolumn{1}{c|}{$\frac{\lVert F\rVert_{F}}{\sqrt{NT}}$} &
\multicolumn{1}{c|}{$\frac{\lVert M\rVert_{F}}{\sqrt{NT}}$} &
\multicolumn{1}{c|}{$\sqrt{\frac{\operatorname{tr}(\Sigma)}{T}}$} &
\multicolumn{1}{c|}{AR(2)} &
\multicolumn{1}{c|}{$(\lambda_{unit},\lambda_{time},\lambda_{nn})$} &
\multicolumn{6}{c|}{RMSE} \\
\hline
 &  &  &  &  &  &
\multicolumn{1}{c|}{TROP} &
\multicolumn{1}{c|}{SDID} &
\multicolumn{1}{c|}{SC} &
\multicolumn{1}{c|}{DID} &
\multicolumn{1}{c|}{MC} &
\multicolumn{1}{c|}{DIFP} \\
\hline
\multicolumn{12}{|l|}{\textbf{Panel A: CPS Data} \quad (Y: log-wage; W: min wage)}\\
\hline
Baseline        & 0.992 & 0.10  & 0.098 & $(0.010,-0.057)$ & $(0,0.1,0.9)$       & {\bfseries 0.025} & 0.029 & 0.037 & 0.050 & 0.032 & 0.032 \\
\hline
\multicolumn{12}{|l|}{\textit{Simpler outcome models}}\\
\hline
No Auto Corr    & 0.992 & 0.10  & 0.098 & $(0,0)$          & $(0.7,0.25,0.6)$    & {\bfseries 0.027} & 0.029 & 0.037 & 0.049 & 0.033 & 0.032 \\
No $M$          & 0.992 & 0     & 0.098 & $(0.010,-0.057)$ & $(0.1,0.025,0.121)$ & 0.015             & 0.018 & 0.018 & {\bfseries 0.013} & 0.015 & 0.017 \\
No $F$          & 0     & 0.10  & 0.098 & $(0.010,-0.057)$ & $(1.4,0.25,0.301)$  & 0.026             & 0.029 & {\bfseries 0.022} & 0.049 & 0.032 & 0.032 \\
Only Noise      & 0     & 0     & 0.098 & $(0.010,-0.057)$ & $(1.8,0.005,0.4)$   & 0.015             & 0.018 & 0.014 & {\bfseries 0.013} & 0.015 & 0.017 \\
\hline
\multicolumn{12}{|l|}{\textbf{Panel B: PWT Data} \quad (Y: log-GDP; W: democracy)}\\
\hline
Baseline        & 0.972 & 0.229 & 0.069 & $(0.913,-0.221)$ & $(0.3,0.4,0.006)$   & {\bfseries 0.023} & 0.036 & 0.040 & 0.198 & 0.046 & 0.039 \\
\hline
\multicolumn{12}{|l|}{\textit{Simpler outcome models}}\\
\hline
No Auto Corr    & 0.972 & 0.229 & 0.069 & $(0,0)$          & $(0.2,0.3,0.016)$   & {\bfseries 0.027} & 0.030 & 0.033 & 0.197 & 0.046 & 0.033 \\
No $M$          & 0.972 & 0     & 0.069 & $(0.913,-0.221)$ & $(0,0.12,0.051)$    & 0.015             & 0.016 & 0.019 & 0.016 & {\bfseries 0.014} & 0.018 \\
No $F$          & 0     & 0.229 & 0.069 & $(0.913,-0.221)$ & $(1.6,0.2,0.041)$   & {\bfseries 0.028} & 0.036 & {\bfseries 0.028} & 0.198 & 0.046 & 0.039 \\
Only Noise      & 0     & 0     & 0.069 & $(0.913,-0.221)$ & $(2,0.1,0.142)$     & 0.015             & 0.016 & 0.016 & 0.016 & {\bfseries 0.014} & 0.018 \\
\hline
\end{tabular}
}
\caption{\textit{Simple outcome models with CPS log-wage (Panel A) and PWT log-GDP (Panel B) data:} RMSEs of six estimators for semi-synthetic simulations with data generated using CPS, PWT, Germany reunification, Basque, Smoking and Boatlift data as described in Section \ref{sec:numerical_illustration}. TROP refers to our proposed estimator, SDID to the Synthetic Difference in Differences, SC to Synthetic Control, DID to Difference in Diferences, MC to Matrix Completion, and DIFP to the SC estimator after recentering its mean ({\it i.e.,} with an intercept). RMSE calculations are out-of-sample (the square-root of the average of the squared errors at the unit/period level) and based on 1000 simulation replications. The simulation design is the same as that used in Table 2 in \cite{arkhangelsky2019synthetic} and formally described in Section \ref{sec:numerical_illustration}.  Specifically, the treatment is simulated by estimating the probability of treatment as in Equation \eqref{eqn:treatment_assignment} using data from the treatment variable.   The outcome and treatment variable used in simulations is log-wage for CPS data and log-GDP for PWT data, and treatments are minimum wage for CPS and democracy for PWT data. 
Columns 2 - 4 defines features of the generating process for the outcome estimated as described in Section \ref{sec:design}. $(\lambda_{unit}, \lambda_{time}, \lambda_{nn})$ denote the estimated tuning parameters for TROP, estimated via cross-validation. Here we consider ten treated periods and ten treated units.  }
\label{tab:sim_shutdown}
\end{table}

%It is important to note that larger $T$ or $N$ do not necessarily improve DID estimators with interactive fixed effects. This is because the parallel trend assumption required by DID is a feature of the data-generating process and not a function of $N$ or $T$. Therefore, it is important to ask whether including $M$ in the design of the data-generating process is crucial to approximate real-world datasets (both for large or small $N,T$).  

Given the  sensitivity of DID  to the presence of interactive components, we next investigate how important interactive fixed effects are  in empirical relevant settings. 
To study this question, we document how much residual variation is explained by an interactive fixed effect model over  a TWFE model. We report results in Table \ref{tab:comparisons_interactive}. Across all our datasets, adding even a single additional interactive fixed effect can significantly decrease the root-mean-squared error (RMSE)  between $10$ and $60\%$.

We conduct a formal test to determine whether to include an interactive fixed effect as follows: for each application, we divide the panel into two halves based on the total number of observations, $T$. We then test whether unit-fixed effects (for each unit) estimated in the first half of observations are the same as in the second half. Rejection of this hypothesis is supportive of interactive fixed effects. We reject this hypothesis for more than 58\% of units in each of these applications. This provides further evidence of the importance of interactive fixed effects across all these applications. 

We complement these tests by reporting the bias in simulations corresponding to failing to account for interactive fixed effects via a DID estimator. 
Figure \ref{fig:bias_DID} shows the bias of a DID estimator across 1000 simulation replications as the treatment is generated by the logistic assignment with a single-factor model. The figure shows that selection into treatment can strongly depend on an unobserved interactive component, in which case the DID may suffer a large bias. 

Overall, these results are suggestive that, although methods such as DID can perform well in the absence of interactive effects, interactive effects are important in applications and can introduce large misspecification bias for models that ignore interactive fixed effect components. 

Finally, as we remove the fixed effect component ($F$ in Equation \eqref{eqn:F_and_M}), Table \ref{tab:sim_shutdown} shows that the Synthetic Control method performs well and similar (or slightly better) to the TROP estimator. This is suggestive that the advantage of the TROP estimator is to jointly capture interactive and fixed effect components through weighting and regression adjustments. 

\begin{table}[H]
\captionsetup{font=scriptsize}
\resizebox{\textwidth}{!}{
\begin{centering}
\scalebox{0.15}
{
\begin{tabular}{c|c|c|c|c|c}
\multirow{2}{*}{dataset} & \multirow{2}{*}{$(N,T)$} & \multicolumn{3}{c|}{RMSE} & \multirow{2}{*}{rejection \% at 5\% level}\tabularnewline
\cline{3-5}
 &  & TWFE & plus factor & \% decrease & \tabularnewline
\hline 
CPS logwage & $(50,40)$ & 0.13 & 0.10 & 23\% & 76\%\tabularnewline
\hline 
CPS unemp-rate & $(50,40)$ & 0.70 & 0.62 & 11\% & 58\%\tabularnewline\hline 
PennWT & $(111,48)$ & 0.25 & 0.11 & 56\% & 91\%\tabularnewline
\hline 
Germany & $(17,44)$ & 0.14 & 0.06 & 57\% & 88\%
\tabularnewline
\hline  
Basque & $(18,43)$ & 
0.19 & 0.08 & 58\% & 100\%
\tabularnewline  
\hline  
Smoking & $(39,31)$ & 
0.37 & 0.20 & 46\% & 79\%
\tabularnewline  
\hline  
Boatlift& $(44,19)$ & 
0.48 & 0.39 & 19\% & 66\%
\tabularnewline  
%WAGEPAN (Vella and Verbeek, 1998) & $(545,8)$ & 0.80 & 0.67 & 16\% & 36\%\tabularnewline\hline 
%County Voting (Gentzkow et al., 2011) & $(1531,32)$ & 0.61 & 0.47 & 23\% & 76\%\tabularnewline
\end{tabular}
}
\par\end{centering}
}
\caption{Out of sample RMSE of TWFE and TWFE plus a leading factor estimated from the data. The column ``$\%$ decrease'' indicates the percentage decrease in RMSE by introducing a single-factor for predicting the outcome across the different applications. The last column indicates the percentage
of rejections of hypothesis that unit fixed effects do not vary in
the first vs. second halves of time periods. These results are suggestive of the relevance of a factor model to better approximate the outcomes in the applications considered in our main simulations. } \label{tab:comparisons_interactive}

\end{table}

\begin{figure}[!ht] 
\captionsetup{font=scriptsize}
    \centering 
        \includegraphics[scale = 0.7]{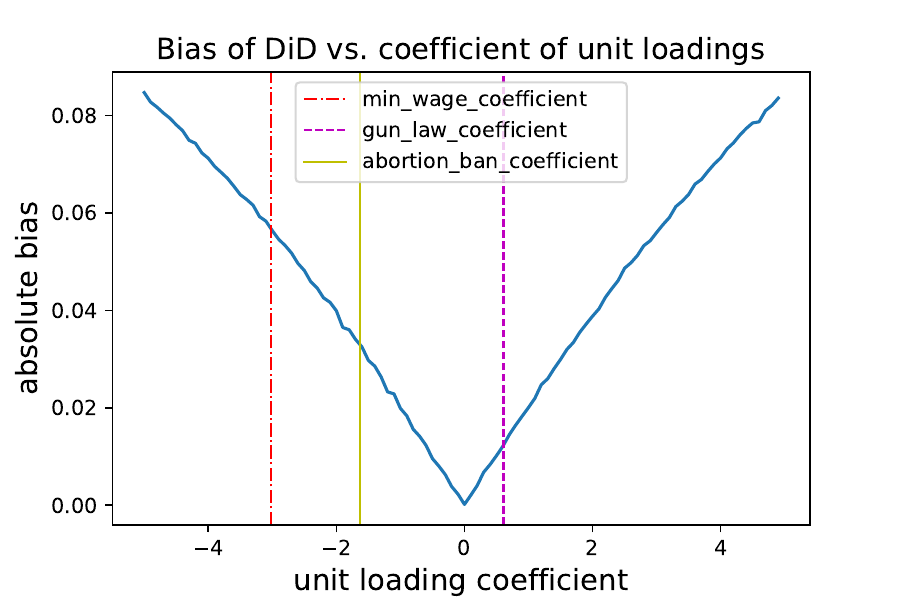}%
   \caption{Bias of DID estimator on CPS data obtained through 10000 replications where the treatment assignment mechanism follows a logistic model as a function of unobserved factor loadings as described in Equation \eqref{eqn:treatment_assignment}.} \label{fig:bias_DID} 
    \end{figure}

\subsection{Varying $T_{\rm pre}$, $N_{\rm co}$, $T_{\rm tr}$, and $N_{\rm tr}$}
%\guido{this is great. All four figures should be on the same axis. I also want to think about a way of making explicit that on the far right all four figures agree. May be add dots there in the graph and in the discussion give the numbers there? }

In this section, we explore the role that the sample size of the different subsamples plays in the performance of each estimator. 
For the sake of brevity, we focus on the PWT data with the treatment and outcome generated as described in Section \ref{sec:design} using for the outcome log-GDP and for the treatment democracy. 

In the baseline simulation, we have $N_{\rm tr}=10$ treated units, $N_{\rm co}=101$ control units, $T_{\rm post}=10$ post-treatment periods, and $T_{\rm pre}=38$ pre-treatment periods. In this baseline case, the RMSE of the six estimators is 2.5\% (TROP), 3.6\% (SDID), 4\% (SC), 19.8\% (DID), 4.6\% (MC), and 3.9\% (DIFP) (in percentages of the variance of the outcome).

We vary $N_{\rm co}$, $N_{\rm tr}$, $T_{\rm post}$, and  $T_{\rm pre}=38$, one at a time, to see how the relative performance of the estimators changes. In each case, we plot the RMSE of the six estimators (TROP, SDID, SC, DID, MC, and DIFP) as a function of the single  feature we are varying. In all cases, the plot includes the baseline case with
$N_{\rm tr}=10$, $N_{\rm co}=101$, $T_{\rm post}=10$, and $T_{\rm pre}=38$. This case is always on the far right of the four figures. 

For the first plot in Figure \ref{fig:vary_T} (left-panel), we vary the number of control units between 10 to 101. For a fixed value of $N_{\rm co}$, we randomly select $N_{\rm co}$ control units from the set of $101$ that are available in the full sample. On the right panel in Figure \ref{fig:vary_T}, we replicate the same exercise where we increase the number of pretreatment time periods from 5 to 38. 
In Figure \ref{fig:T_train}, we vary the number of treated units from 1 to 10 (the left panel) and the number of treated periods from 1 to 10 (the right panel).

For a small number of control units, we find that SDID, DIFP, and TROP perform similarly, while MC, SC, and DID perform poorly. As the number of control units increases, TROP dominates all other competitors considered. The DID estimator benefits slightly from more control units (not surprisingly as two of the four averages in that estimator will now be estimated more precisely), but not much because the increase in the number of control units does not improve the bias. The estimators that assign different weights to different units benefit more than the DID estimator because the increase in the number of control units will remove some of the bias (conditional on the unit selected for treatment).

Regarding the number of pre-treatment periods, we find that methods are mostly comparable for $T_{\mathrm{pre}} \le 10$; after this TROP illustrates a significan improvement with respect to SC and other estimators. 

Here, increasing the number of pretreatment periods has multiple effects. One effect is increasing the sample size and therefore improving precision. However, with the treated periods always at the end of the panel, the bias for estimators that do not allow for variation in the time weights is likely to increase. We see that clearly with the DID estimator. The other estimators do not deteriorate as much, but all do to some extent, suggesting that further research into optimal time weights beyond our simple exponential decline might be productive.

To further explore the behavior of the estimators, in Figure \ref{fig:T_train} we vary the number of treated units from one to ten and similarly the number of treatment periods. 
Increasing the number of treated units slightly improves the DID estimator, but leads to a much bigger improvement for the estimators that select the control units more deliberately, and leads to a larger relative improvement in the TROP estimator compared to the other estimators.
Increasing the number of treatment periods, with these periods always coming after the pre-treatment periods, leads to a deterioration in the performance of all the estimators, but the relative performance of the TROP estimator improves.

\begin{figure}[!ht]
\captionsetup{font=scriptsize}
    \centering 
        \includegraphics[scale = 0.5]{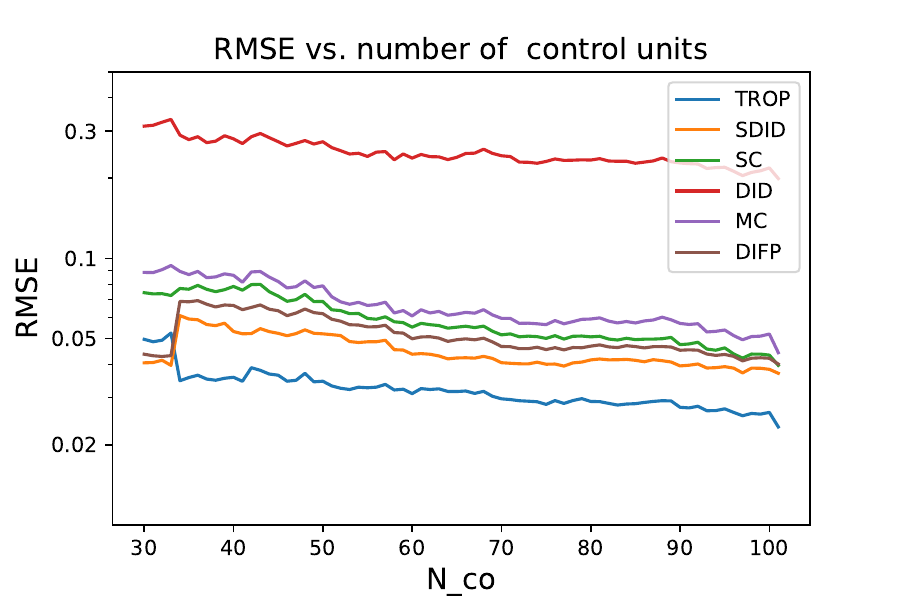}
 \includegraphics[scale = 0.5]{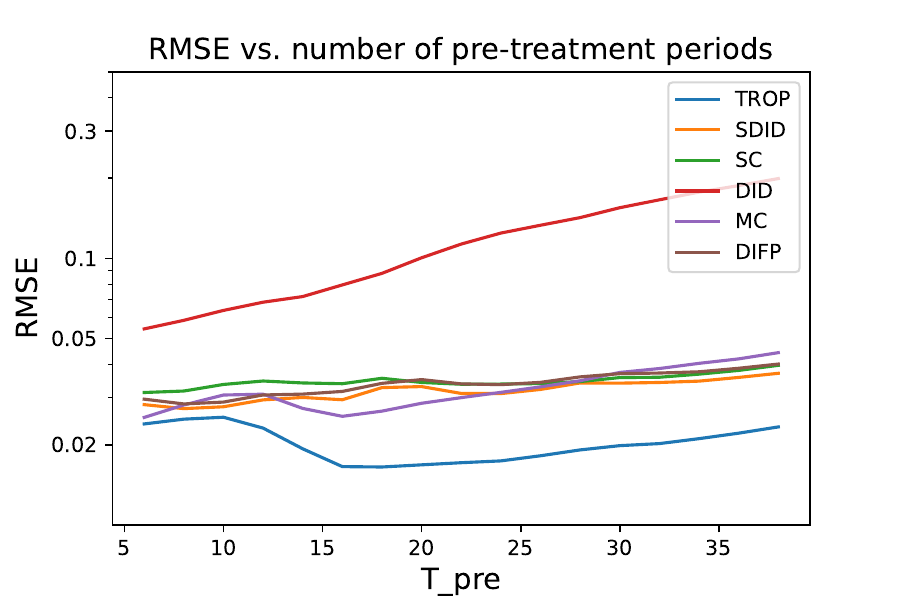}
   \caption{RMSEs of estimators on the PWT dataset  with the outcome and treatment generated as described in Section \ref{sec:design}, using as the outcome the log-GDP as the outcome and as the treatment democracy. In the figure, we vary the number of control  units (left) or the number of pre-treatment periods (right). In both cases, the number of control periods and units is ten, with the treatment periods corresponding to the last ten periods in the simulated panel. When we vary the number of control units, we select randomly $N$ units from the full PWT panel while keeping the total number of period $T=48$. When we select the number of pre-treatment periods, we construct shorter panels by selecting the last $T$ periods of the full PWT panel while keeping $N=111$. } \label{fig:vary_T}
    \end{figure}

\begin{figure}[!ht]
\captionsetup{font=scriptsize}
    \centering 
         \includegraphics[scale = 0.5]{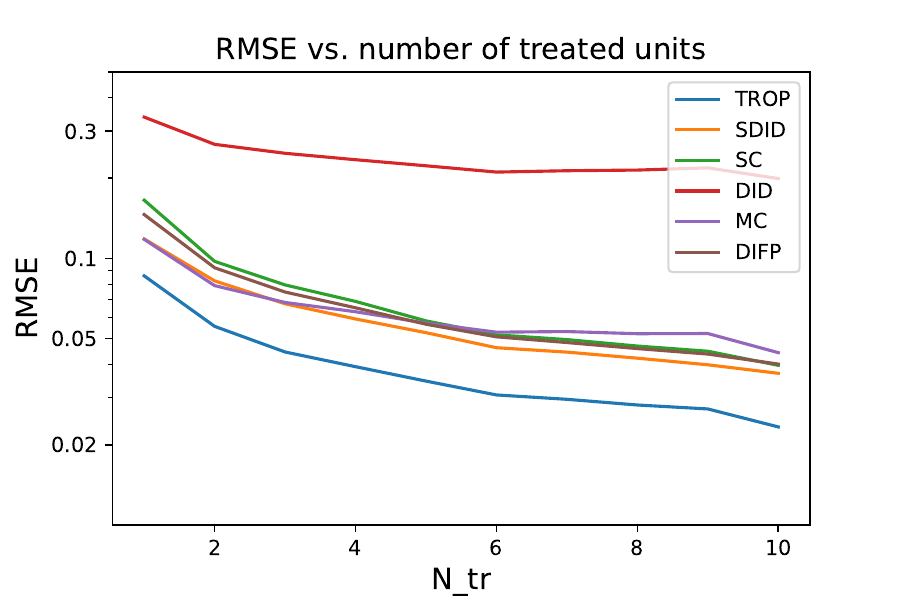}
\includegraphics[scale = 0.5]{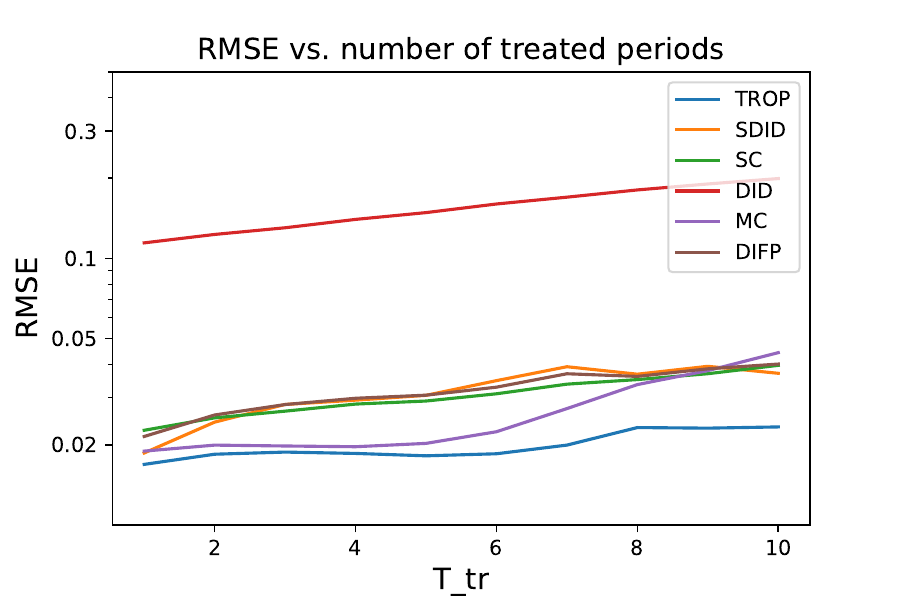}
       
   \caption{RMSEs of estimators on the PWT dataset  with the outcome and treatment generated as described in Section \ref{sec:design}, using as the outcome the log-GDP as the outcome and as the treatment democracy.
   Here we vary the number of treated units (left) or the number of treatment periods units (right). For the panel on the left, we use $N_{\rm co}=101$ and $T_{\rm co}=38$, while varying $N_{tr}=1,\dots,10$ with $T_{tr}=10$. For the panel on the right, we use $T_{\rm co}=38$ and $N_{\rm co}=111$, while varying $T_{tr}=1,\dots,10$ with $N_{tr}=10$. } \label{fig:T_train}
    \end{figure}

\subsection{Critical Features of the TROP Estimator} 

To better understand what drives the strong performance of TROP, we consider several variants where we shut down components of the TROP estimator. This includes shutting down the factor model (selecting $\lambda_{nn} = \infty$), or the time weights
(selecting $\lambda_{time} = 0$)
or the unit weights (selecting $\lambda_{unit} = 0$) or combinations of those.

Table \ref{tab:turn_off_components_relative_numeric} illustrates the results across the six main simulation designs. As in our main specifications, we focus on CPS data using treatment generated from minimum wage and unemployment rate, PWT dataset using democracy for treatment, and for the remaining four datasets, we consider treatment assigned randomly across the units. We consider  ten treated units and periods (Appendix Table \ref{tab:one_shutdown} reports results for one treated unit and period with similar results to the ones described here).   
We report the RMSE relative to the original TROP estimator across the different dataset (Appendix Table \ref{tab:RMSE_and_bias_shutdown} reports the RMSE in absolute term and bias of the estimator).

For ten treated units and post-treatment periods, we find that removing the regression adjustment $(\lambda_{nn} = \infty$) or the time weights ($\lambda_{time} = 0$) increases significantly the RMSE, up to $85\%$ for $\lambda_{nn} = \infty$ and $91\%$ for $\lambda_{time} = 0$. Interestingly, removing unit weights $(\lambda_{unit} = 0$) increases RMSE only in some but not all cases, with a maximum increase in RMSE $10\%$.  
Removing both $\lambda_{unit}$ and $\lambda_{nn}$ or both $\lambda_{unit}$ and $\lambda_{time}$ can significantly increase RMSE to $90\%$. Setting $\lambda_{time} = 0$ and $\lambda_{nn} = \infty$
 can have the largest effect on RMSE  among all combinations of two-out-of-three parameters, with an increase in RMSE up-to five times greater.  

 At the same time, we observe substantial heterogeneity across the different applications, where the different weights or regression adjustments can play different roles depending on the particular application. For one treated unit and period (Panel B in Table \ref{tab:turn_off_components_relative_numeric}) we similarly find large heterogeneity in the role played by each of the three components. (Also, for a few specifications with one treated unit and period, TROP with no constraints slightly underperforms by one to three percentage points to introducing some constraints due to approximation error.) 

In summary, Table \ref{tab:turn_off_components_relative_numeric} highlights the key advantage of the TROP of learning from the data how to weight observations across units and time, combined with a flexible predictor model: failing to include time weights and regression adjustment together, for example, may more than double RMSE in some applications, while in others a similar drop in performance occurs when failing to use time and unit weights. 

\begin{table}[H]
\captionsetup{font=scriptsize}
\centering
\scalebox{0.6}{
\begin{tabular}{|l|c|c|c|c|c|c|c|}
\hline
\multicolumn{1}{|l|}{\textbf{Panel A: } $T_{\mathrm{post}}, N_{\mathrm{tr}} = 10$} & CPS logwage & CPS urate & PWT & Germany & Basque & Smoking & Boatlift \\
\hline\hline
TROP                                        & 1.00 & 1.00 & 1.00 & 1.00 & 1.00 & 1.00 & 1.00 \\
\hline
$(\lambda_{nn}=\infty)$                     & 1.00 & 1.21 & 1.65 & 1.48 & 1.85 & 1.46 & 1.55 \\
$(\lambda_{unit}=0)$                        & 1.00 & 1.10 & 1.04 & 1.00 & 1.00 & 1.04 & 1.00 \\
$(\lambda_{time}=0)$                        & 1.28 & 1.11 & 1.91 & 1.28 & 1.51 & 1.20 & 1.05 \\
$(\lambda_{unit}=\lambda_{time}=0)$         & 1.28 & 1.10 & 1.91 & 1.28 & 1.51 & 1.20 & 1.05 \\
$(\lambda_{unit}=0,\lambda_{nn}=\infty)$    & 1.00 & 1.22 & 1.70 & 1.48 & 1.85 & 1.46 & 1.55 \\
$(\lambda_{time}=0,\lambda_{nn}=\infty)$    & 1.92 & 1.78 & 5.35 & 2.84 & 3.54 & 2.26 & 1.36 \\
\makecell{DID\\ $(\lambda_{unit}=\lambda_{time}=0,\lambda_{nn}=\infty)$}
                                            & 1.96 & 1.89 & 8.61 & 2.92 & 3.54 & 2.26 & 1.36 \\
\hline
\multicolumn{8}{|l|}{\textbf{}}\\[-0.9ex]
\hline
\multicolumn{1}{|l|}{\textbf{Panel B: } $T_{\mathrm{post}}, N_{\mathrm{tr}} = 1$} & CPS logwage & CPS urate & PWT & Germany & Basque & Smoking & Boatlift \\
\hline\hline
TROP                                        & 1.00 & 1.00 & 1.00 & 1.00 & 1.00 & 1.00 & 1.00 \\
\hline
$(\lambda_{nn}=\infty)$                     & 1.00 & 0.99 & 1.00 & 1.00 & 1.00 & 1.02 & 1.13 \\
$(\lambda_{unit}=0)$                        & 1.00 & 1.03 & 1.03 & 1.00 & 1.09 & 1.08 & 1.12 \\
$(\lambda_{time}=0)$                        & 1.01 & 0.99 & 1.15 & 1.51 & 4.03 & 1.15 & 1.01 \\
$(\lambda_{unit}=\lambda_{time}=0)$         & 1.04 & 1.09 & 1.15 & 1.51 & 1.31 & 1.17 & 1.00 \\
$(\lambda_{unit}=0,\lambda_{nn}=\infty)$    & 1.01 & 0.99 & 1.03 & 1.00 & 1.09 & 1.08 & 1.14 \\
$(\lambda_{time}=0,\lambda_{nn}=\infty)$    & 1.23 & 0.99 & 9.85 & 3.08 & 6.51 & 3.38 & 1.42 \\
\makecell{DID\\ $(\lambda_{unit}=\lambda_{time}=0,\lambda_{nn}=\infty)$}
                                            & 1.41 & 0.99 & 11.06 & 3.68 & 4.37 & 3.61 & 1.38 \\
\hline
\end{tabular}
}
\caption{Simulation designs as in Table \ref{tab:first}, where we report the TROP estimator (first row) and the TROP estimator constraining one or combinations of $\lambda_{time}, \lambda_{unit}$ to be zero (corresponding to no time or unit weights) and $\lambda_{nn}$ to be infinity (corresponding to no regression adjustment). Each column reports the RMSE relative to TROP within each dataset column (lower is better; TROP = 1.00).
 The number of treated units is ten and the number of treated periods is ten for Panel A and one for each in Panel B. The treatment simulated as in Equation  \eqref{eqn:treatment_assignment} is calibrated for CPS log-wage and CPS unemployment rate to minimum wage, for PWT to the democracy indicator, and for Germany, Basque, Smoking and Boatlift is randomly assigned.}
\label{tab:turn_off_components_relative_numeric}
\end{table}

Table \ref{tab:weights2} further illustrates key features of the TROP estimator with ten treated periods and units, showing how the cumulative weight between units and time periods changes between the different datasets. Across all applications more than $50\%$ of the weight is assigned to the closest five periods, while unit weights are less concentrated in applications. This suggests the relevance of time weighting in particular to improve the performance in realistic applications.  
 
{\small
\begin{table}[H]
\captionsetup{font=scriptsize}
\resizebox{\textwidth}{!}{
\begin{centering}
\scalebox{0.15}
{
\begin{tabular}{llccccc}
 %& \multicolumn{6}{c}{RMSE}\tabularnewline\hline 
 && &&\multicolumn{3}{c}{Cumulative Weight}\\
Dataset& treatment & $N$ & $T$ & last 5 periods &  closest 5 units & closest half units \tabularnewline
\hline 
\hline 
CPS logwage& min wage & 50 & 40 & 0.56 & 0.13 & 0.59\tabularnewline
CPS unemp-rate& min wage & 50 & 40& 0.66 & 0.23 &0.67\tabularnewline
%\hline 
%CPS hours& min wage& 50 & 40 & 0.63 & 0.24 &0.78\tabularnewline
%\hline 
%CPS log-wage& gun law& 50 & 40 & 0.55 & 0.15 & 0.65 \tabularnewline
%\hline 
%CPS log-wage& abortion & 50 & 40& 0.67 & 0.14 & 0.61\tabularnewline
%\hline 
%CPS log-wage& random & 50 & 40& 0.55 & 0.15 & 0.65\tabularnewline
%\hline 
PWT log-GDP& democracy & 111 & 48 & 0.70 & 0.06 & 0.59\tabularnewline
%\hline 
%PWT log-GDP& education & 111 & 48  & 0.55 & 0.22& 0.75\tabularnewline
%\hline 
%PWT log-GDP& random & 111 & 48  & 0.55 & 0.20 & 0.85 \tabularnewline
Germany & random & 17 & 44 & 0.50 & 0.37 & 0.55  \tabularnewline
Basque & random & 18 & 43  &0.66  &0.28  &0.50  \tabularnewline
Smoking & random & 39& 31 &0.70  &  0.15&  0.56\tabularnewline
Boatlift & random & 44& 19 &0.51  & 0.13  &  0.54\tabularnewline
\end{tabular}
}
\par\end{centering}
}
\caption{Cumulative weights across time or units of the TROP estimator based on the cross-validated tuning parameters across applications. The number of treated units and the number of treated periods are both 10.} \label{tab:weights2}
\end{table}
}

Finally, Appendix Tables \ref{tab:two_panels_Tpost_Ntr}, \ref{tab:sim_bias_only}, and \ref{tab:one_shutdown} report RMSE and bias both across estimators and as we “shut down’’ different components of TROP. Two main patterns emerge. First, there is no single uniformly best competitor: on CPS data, TROP typically has smaller bias than SC, whereas on PWT data the reverse can occur, while DID exhibits substantially larger bias than all other methods in all designs. Second, for TROP the bias component is quantitatively modest relative to its variance—its squared bias is often an order of magnitude smaller than the variance term—and it remains small even in designs where other estimators are heavily biased. When bias is negligible for all methods, TROP therefore improves performance mainly by reducing variance, delivering lower RMSE at essentially no cost in bias; when bias is large for SC, DID, or MC, TROP both lowers RMSE and sharply reduces bias, leading to sizeable MSE gains.
\section{Formal analysis} \label{sec:formal}
In this section we provide an intuitive description of properties of our estimator. To facilitate our comparison with existing procedures, consider estimating the counterfactual in a setting where  $W_{i,t} = 1\{i >  N_0\} 1\{t > T_0\}$ for fixed $(N_0,T_0)$. Let $T_1 = T - T_0 = 1, N_1 = N - N_0 = 1$. Therefore,  we have a single treated unit and treatment period. We denote the weights as 
$$
\omega \in \mathbb{R}^{N_0}, \theta \in \mathbb{R}^{T_0}
$$
denoting unit and time weights, where we suppressed their dependence with the index of the treated unit and period given our focus on a single treated unit (the analysis extends to multiple treated units and treatment periods, omitted for expositional convenience). 
 
 We will discuss theoretical properties for \textit{given} weights, interpreting these as probability limits of estimated weights similar to \cite{arkhangelsky2019synthetic} and \cite{imbens2023identification}). As \cite{arkhangelsky2019synthetic} show, assuming deterministic weights is almost without loss whenever $N_0 \gg N_1, T_0 \gg T_1$, that is, whenever $N_0$ and $T_0$ are sufficiently larger than the error from the post-treatment period and treated units. This is because the estimation error of the weights depends on $N_0$ and $T_0$ whereas the estimation error of the counterfactual also depends on $N_1$ and $T_1$ typically much smaller. Therefore, we will focus on settings where the weights converge to \textit{some} probability limit, with their corresponding estimation error being negligible relative to the error of the post-treatment idiosyncratic shocks, as we further discuss in Section \ref{sec:inference}. The weights sum up to one, so that $1^\top \omega = 1^\top \theta = 1$. 
 
\begin{ass} \label{ass:factor_model} Suppose that the following holds: for all $(i,t)$
\begin{itemize}
\item[(i)] $W_{i,t} = 1\{i >  N_0\} 1\{t > T_0\}$ for some constants $(N_0, T_0)$. 
\item[(ii)] for some $\Gamma_i, \Lambda_t \in \mathbb{R}^K$ (for arbitrary $K$) 
$$
Y_{i,t}(0) = \mathbf{L}_{i,t} + \varepsilon_{i,t}, \quad \mathbf{L}_{i,t} = \Gamma_i^\top \Lambda_t, \quad \mathbb{E}[\varepsilon_{i,t} | \mathbf{L}] = 0.  
$$
\end{itemize} 
Define $\Gamma = \Big[\Gamma_1, \cdots, \Gamma_N\Big]^\top$ and $\Lambda = \Big[\Lambda_1, \cdots, \Lambda_T\Big]^\top$, so that $\mathbf{L} = \Gamma \Lambda^\top$. 
\end{ass} 
Assumption \ref{ass:factor_model}(ii) imposes a factor model structure on the outcome of interest with $K$ being arbitrary (possibly growing with the sample size). Note that for arbitrary factor models, $\Gamma$ and $\Lambda$ are defined up-to a rotation. This will not affect our analysis given that our goal is to predict $\mathbf{L}_{i,t}$, and not separately each factor. 
Finally, because here we consider general estimators, without loss (and with a slight abuse of notation) we directly incorporate unit and time fixed effects as part of $\Gamma_i, \Lambda_t$ (assuming that the first entry is constant for $\Lambda_t$ and the second entry is constant for $\Gamma_i$). 

For our analysis, we introduce some notation discussed in Table \ref{tab:methods3}. 

\begin{table}[!ht] \centering 
  \caption[Caption for LOF]{Notation. } 
  \label{tab:methods3} 
  \scalebox{0.9}{
\begin{tabular}{@{\extracolsep{5pt}} ccc} 
\\[-1.8ex]\hline 
\hline \\[-1.8ex] 
Description & Mathematical Formulation  \\ 
\hline \\[-1.8ex] 
Average Control Loadings 
 &  $\bar{\Gamma}_0 = \frac{1}{N_0} \sum_{j \le N_0} \Gamma_j$  \\ 
& &  
 \\
%Average Treatment Loadings & $\bar{\Gamma}_1 = \frac{1}{N_1} \sum_{j > N_0} \Gamma_j$ \\ 
% & & \\ 
 Average Control Factors 
 &  $\bar{\Lambda}_0 = \frac{1}{T_0} \sum_{s \le T_0} \Lambda_s$  \\ 
& &  
 \\
%Average Treatment Factors & $\bar{\Lambda}_1 = \frac{1}{T_1} \sum_{s > T_0} \Lambda_s$ \\ 
% & & \\ 
 Weighted Average Control Loadings 
 &  $\bar{\Gamma}_0(\omega) = \sum_{j \le N_0} \omega_j \Gamma_j$  \\ 
 && \\ 
 Weighted Average Treatment Loadings 
 &  $\bar{\Lambda}_0(\theta) = \sum_{s \le T_0} \theta_s \Lambda_s$  \\ 
 && \\
 Unit Discrepancy & $ \Delta^{\mathrm{u}}(\omega, \Gamma)= \bar{\Gamma}_0(\omega) - \Gamma_N$ \\ 
 && \\ 
 Time Discrepancy & $\Delta^{\mathrm{t}}(\theta, \Lambda) = \bar{\Lambda}_0(\theta) - \Lambda_T$ \\ 
 && \\ 
\hline \\[-1.8ex] 
\end{tabular} 
}
\end{table}

%In addition, define 
%$$
%\begin{aligned} 
%\tau(0) = \bar{\Lambda}_1^\top  \bar{\Gamma}_1, \quad \bar{\varepsilon}^1 & = \frac{1}{T_1 N_1}  \sum_{j > N_0, t > T_0} \varepsilon_{jt},  
%\end{aligned} 
%$$ 
%the expected potential outcome under control (on the treated unit). 

\subsection{TROP as a balancing estimator and triple robustness}

In the single treated unit and period case, the estimated ATT reads as 
\begin{equation} \label{eqn:dr_representation} 
\small 
\begin{aligned} 
\hat{\tau} & = Y_{NT} - \hat{\mu}^{\mathrm{TROP}}(0, \theta, \omega) \\ 
\hat{\mu}^{\mathrm{TROP}}(0, \theta, \omega) & := \widehat{\mathbf{L}}_{NT} +   \left(\sum_{t \le T_0} \theta_t\Big(Y_{Nt} - \widehat{\mathbf{L}}_{Nt}\Big) + \sum_{i \le N_0} \omega_i \Big(Y_{iT} - \widehat{\mathbf{L}}_{iT}\Big) - \sum_{t \le T_0} \sum_{i \le N_0}  \theta_t\omega_i\Big(Y_{it} - \widehat{\mathbf{L}}_{it}\Big) \right)
\end{aligned} 
\end{equation} 
where we defined $\hat{\mu}^{\mathrm{TROP}}(0, \theta, \omega)$ the estimated counterfactual $Y_{NT}(0)$ for the treated unit. That is, we can interpret the estimated counterfactual through the lens of regression adjustments: the weights $\theta, \omega$ serve as balancing weights over time and units respectively. The estimator $\widehat{\mathbf{L}}$ serves as a regression adjustment. Different from previous sections, here $\widehat{\mathbf{L}}$ can be constructed using several estimators of interest (without necessarily using a nuclear norm penalization).

%The following lemma holds. 
%\begin{lem} \label{lem:first} Suppose that %Assumption \ref{ass:factor_model}(i) hold. Let $\tau$ be defined as in Equation \eqref{eqn:tau}. Then 
%$$ 
%\hat{\tau} = \tau + \tau(0) - \hat{\mu}^{\mathrm{TROP}}(0, \theta, \omega) + \bar{\varepsilon}^1.  
%$$ 
%\end{lem} 

%The proof follows directly from the fact that  $\frac{1}{T_1 N_1} \sum_{j > N_0, t > T_0}  Y_{jt} = \tau + \tau(0) + \bar{\varepsilon}^1$. 
It is natural to ask whether, by including an additional regression adjustment, TROP is asymptotically unbiased as $N_0, T_0 \rightarrow \infty$ for $\tau$ even in the absence of perfect balance between units or time. The following theorem shows that the proposed estimator inherits a triple robustness property. 

To establish this result we impose the following restriction on the class of estimators for $\mathbf{L}_{i,t}$.

\begin{ass}[Regression adjustment estimators] \label{ass:class_estimators}  There is a matrix $B$ of dimension $K \times K$ such that 
$$
\mathbb{E}\left[\widehat{\mathbf{L}} | \mathbf{L}\right] = \Gamma (\mathbb{I} + B) \Lambda^\top.
$$ 
where, $\Gamma,\Lambda$ are defined in  Assumption \ref{ass:factor_model} and $\mathbb{I}$ is the identity matrix of dimension $K \times K$.  
\end{ass} 

Assumption \ref{ass:class_estimators} introduces a broad class of regression adjustments. The matrix $B$ summarizes the bias of $\widehat{\mathbf L}$ relative to $\mathbf L$: whenever $B$ is the zero-matrix, $\widehat{\mathbf{L}}$ is an unbiased estimator, whereas $B$ different from the zero matrix captures shrinkage and rank truncation (e.g., if $\operatorname{rank}(\mathbf L)=K'>K$, a rank-$K$ procedure typically yields a non-zero bias $B$). Assumption \ref{ass:estimator_covariates} is satisfied, under isotropic additive noise, by several methods that include PCA/truncated SVD, nuclear-norm penalized least squares, and general singular-value–shrinkage estimators whose predictions rotate equivariantly with the singular vectors of $\mathbf L$
\citep{cai2010svt, candes2013unbiased, gavish2017optimal, negahban2011lowrank, nadakuditi2014optshrink}.\footnote{Formally, let $Y=L+\varepsilon$ with $\operatorname{rank}(L)=K'$ and SVD $L=U\Sigma V^\top$. Noise isotropy implies that for all orthogonal $Q_1,Q_2$, $Q_1\varepsilon Q_2^\top\mid L =_d \varepsilon\mid L$; and (ii) estimator equivariance implies that for all orthogonal $Q_1,Q_2$, the estimator $\hat{\mathbf{L}}$ can be written as a function $F(\cdot)$ of the data such that $F(Q_1YQ_2^\top)=Q_1F(Y)Q_2^\top$.}
%Then $\bar F(L)\equiv\mathbb{E}[F(Y)\mid L]$ satisfies $\bar F(Q_1LQ_2^\top)=Q_1\bar F(L)Q_2^\top$, hence $\bar F(L)=U\,G(\Sigma)\,V^\top$ for some matrix-valued function $G$. Isotropy within singular subspaces forces $G(\Sigma)$ to be diagonal in the $(U,V)$ basis; write $B(\Sigma)=\mathrm{diag}(b_1,\dots,b_{K'})$. If $F$ has rank at most $r$ (e.g., truncated SVD with $r=K$), then $b_j=0$ for $j>r$. Identifying $U\equiv\Gamma$ and $V\equiv\Lambda$ yields the span-invariance condition $\mathbb{E}[\widehat{\mathbf L}\mid \mathbf L]=\Gamma B\Lambda^\top$. For nuclear-norm and general spectral shrinkers, $F$ is the spectral proximal map, so the same conclusion applies.}
 In this formulation, $B$ characterizes the bias of the regression adjustment.\footnote{Also, note that it is possible to relax Assumption \ref{ass:class_estimators} and assume that the expectation equals $\Gamma (\mathbb{I} + B) \Lambda^\top$ up to a small vanishing term, in which case our robustness properties should be interpreted asymptotically as the additional term convergences to zero. We omit this for expositional convenience only.}

%Before doing so, define the residual average idiosyncratic shocks (of mean zero) 
%\begin{equation} \label{eqn:nu} 
%\bar{\varepsilon}^0(\theta, \omega) = \sum_{i \le N_0, t > T_0} \omega_i \varepsilon_{it} + \sum_{t \le T_0, i > N_0} \theta_t\varepsilon_{it} - \sum_{i \le N_0, t \le T_0}  \omega_i \theta_t\varepsilon_{it}.
%\end{equation} 

\begin{thm}[Triple robustness] \label{thm:1} Let Assumption \ref{ass:factor_model} hold. Then for fixed (not data dependent) weights $\theta, \omega$ (and conditional on $T_0, N_0$), 
$$
\small 
\begin{aligned} 
& \Big|\mathbb{E}[\hat{\tau} - \tau|\mathbf{L}]\Big|  \le \underbrace{||\Delta^{\mathrm{u}}(\omega, \Gamma)||_2}_{\text{(i): Units imbalance}} \times \underbrace{||\Delta^{\mathrm{t}}(\theta, \Lambda)||_2}_{\text{(ii): Time imbalance}} \times \underbrace{||B||_{\star}}_{\text{(iii): misspecification of the regression adjustment}}, 
\end{aligned} 
$$ 
where $||\cdot||_2$ denotes the $l_2$-norm and $||\cdot||_{\star}$ denotes the spectral norm and $\mathbb{I}_K$ denotes the identity matrix and $B$ as defined in Assumption \ref{ass:class_estimators}. 
\end{thm} 

\begin{proof} See Appendix \ref{proof:thm:1}.  
\end{proof} 

 Theorem \ref{thm:1} shows that we can decompose the bias of the estimator into \textit{the product} of three main components: (i) unit imbalance, (ii) time imbalance and (iii) bias of the estimator $\mathbf{L}$, expressed through nuclear norm operator in the final expression. The first two components coincide with the time and unit imbalance components studied, for instance, in \cite{imbens2023identification}. The last component captures the misspecification error from estimating the factors and loadings when, e.g., imposing a low-rank regularization. 
As discussed above, here we consider non-data dependent weights, interpreting $\theta, \omega$ as probability limits of the estimated weights \citep[see e.g.][]{hirshberg2021least}.\footnote{Note that because $\tau$ is also random as it depends on potential outcomes, the expectation is both over $\hat{\tau}$ and $\tau$.}

 \subsubsection{Conditions for exact unbiasedness} \label{sec:identification} 

Theorem \ref{thm:1} provides us with insights on the relevant identification conditions to identify $\mathbb{E}[Y_{NT}(0) | \mathbf{L}]$, the target counterfactual of interest. 

\begin{ass}[Conditions for identification of counterfactual] \label{ass:identification_conditions} Suppose that, under the model in Assumption \ref{ass:factor_model}, at least one of the following conditions hold: 
\begin{enumerate} 
\item $\sum_{i \le N_0} \Gamma_i \omega_i = \Gamma_N$ (balance over unit loadings); or 
\item $\sum_{s \le T_0} \theta_s \Lambda_s = \Lambda_T$ (balance over factor loadings); or 
\item $B = \mathbf{0}_K$ where $\mathbf{0}_K$ is a matrix of dimension $K \times K$ with all zero entries  (correct regression adjustment specification).
\end{enumerate}   
\end{ass}

Condition (a) requires that the weights $\omega_i$ guarantee balance between the unobserved control loadings and the loading of the treated unit, similar to the condition required by standard Synthetic Control. Condition (b) imposes the same with respect to time weights and time-varying factors (instead of loadings). When (a) or (b) holds, the factor model in Assumption \ref{ass:factor_model} may not be low-rank, and therefore the restriction on $\varepsilon_{it}$ may be vacuous (since we can approximate $Y_{i,t}$ with any sufficiently high dimensional factor model). Condition (c) requires instead that the estimated factor model is unbiased, so that $B$ equal the zero matrix.  Condition (c) typically holds for a low-rank factor model, where therefore the exogeneity condition on $\varepsilon_{i,t}$ is binding.

\begin{cor}[Exact unbiasedness] \label{cor:main1} Let Assumptions \ref{ass:factor_model},  \ref{ass:class_estimators},  \ref{ass:identification_conditions} hold. Then $\mathbb{E}[\hat{\mu}^{\mathrm{TROP}}(0,\theta,\omega) - Y_{NT}(0)|\mathbf{L}] = \mathbb{E}[\hat{\tau} - \tau|\mathbf{L}] = 0$.  
\end{cor} 

Corollary \ref{cor:main1} guarantees identification of the target counterfactual of interest if either of the three conditions in Assumption \ref{ass:identification_conditions} hold (under the stated, possibly high dimensional factor model).

\begin{rem}[More general estimators] It is possible to consider general estimators instead of imposing restrictions as in Assumption \ref{ass:class_estimators}. 
This is formalized in Theorem \ref{thm:gen-identity} where we characterize the estimation error for general estimators of $\widehat{\mathbf{L}}$ (at the expense of more elaborate notation).
With general estimators, the estimation error depends on generic properties of the estimator. We provide an example  
in Section \ref{sec:extensions} and Appendix \ref{sec:estimator_covariates}, where in addition to the low-rank factor models, we also consider additive covariates multiplying linear coefficients (without imposing restrictions on the particular choice of the linear estimator). We leave to future research a comprehensive analysis of arbitrary non-linear estimators.  
\end{rem}

\subsection{Comparison with existing estimators}

Theorem \ref{thm:1} allows us to directly compare our proposed estimator to both Synthetic Control and Synthetic DID. 
To gain further insight, we illustrate the source of the bias for the DID, Synthetic Control and Synthetic DID estimator under Assumption \ref{ass:factor_model} alone. These estimators (for given weights) take the form 
$$
\small 
\begin{aligned} 
\hat{\tau}^{\mathrm{sc}} & =  Y_{NT} - \sum_{i \le N_0} \omega_i Y_{iT}, \quad \hat{\tau}^{\mathrm{sDID}} & =  Y_{NT} - \left(\sum_{t \le T_0} \theta_tY_{Nt} + \sum_{i \le N_0} \omega_i Y_{iT}  - \sum_{t \le T_0} \sum_{i \le N_0}  \omega_i \theta_tY_{it} \right). 
\end{aligned} 
$$
We interpret the Synthetic Control as a balancing estimator that imposes balance over units, and the Synthetic DID as a balancing estimator across units and time.

As noted in \cite{imbens2023identification}, with $\theta, \omega$ not data dependent for simplicity, we can write 
\begin{equation} \label{eqn:imbalance}
\begin{aligned} 
B^{\mathrm{DID}} & := \mathbb{E}[\hat{\tau}^{\mathrm{DID}} - \tau | \mathbf{L}]  & = \Big(\Gamma_N - \bar{\Gamma}^0\Big)^\top \Big(\Lambda_T - \bar{\Lambda}^0\Big)    ,  \\ 
B^{\mathrm{SC}} & := \mathbb{E}[\hat{\tau}^{\mathrm{sc}} - \tau | \mathbf{L}]  & = \Big(\bar{\Gamma}_0(\omega) - \Gamma_N \Big)^\top \Lambda_T, \\ B^{\mathrm{SDID}} & := \mathbb{E}[\hat{\tau}^{\mathrm{sDID}} - \tau | \mathbf{L}] & = \Big(\bar{\Gamma}_0(\omega) - \Gamma_N \Big)^\top \Big(\bar{\Lambda}_0(\theta) - \Lambda_T \Big).  
\end{aligned} 
\end{equation}

That is, the bias of the DID estimator depends on the imbalance between the average loadings on the treated and control group, times the imbalance between the factor in the treatment and control period. 
The bias of the Synthetic control estimator depends on imbalances between the loadings of the treated and control, \textit{reweighted} by unit-specific weights (which typically are estimated to minimize such imbalance). The bias of Synthetic DID is doubly robust to imbalance either across time periods or across treated and control units. In particular, both the Synthetic Control and Synthetic DID leverage possibly non-uniform weights to decrease the imbalance between the treated and control units. The Difference-in-Differences estimator instead suffer a bias due to using uniform weights over the control units and control periods.  

Different from all such estimator, our proposed procedure inherits robustness not only from the choice of the weights, but also through the additional regression adjustment. The combination of these three components guarantee a triple robustness properties to the TROP not encountered in previous estimators.

\subsection{Inference: discussion} \label{sec:inference}

In the framework with a single treated unit, with idiosyncratic shock $\varepsilon_{NT}$, the estimation error of the weights is typically negligible (i.e., $N_0 \gg N_1, T_0 \gg T_1$), and as a result inference is driven by the variance of $Y_{NT}$. In the presence of multiple post-treatment periods, as we define $Y_{NT}$ as the average outcome of the treated unit over such periods, we may assume that (keeping $T_0$ and $N_0$ as deterministic) 
$$
\varepsilon_{NT}(1):= Y_{NT}(1) - \mathbb{E}[Y_{NT}(1) | \mathbf{L}] \sim \mathcal{N}(0, \sigma_{NT}^2) 
$$ 
where the normal approximation holds from a central limit theorem argument assuming that the last period denotes the average over multiple post treatment periods. In particular, define 
\begin{equation} \label{eqn:epsilon_zero}
\begin{aligned} 
\bar{\varepsilon}^0(\theta, \omega) & = \sum_{i \le N_0, t > T_0} \omega_i \varepsilon_{it} + \sum_{t \le T_0, i > N_0} \theta_t\varepsilon_{it} - \sum_{i \le N_0, t \le T_0}  \omega_i \theta_t\varepsilon_{it}.  
\end{aligned} 
\end{equation} 
It follows that 
$$
\hat{\tau} - \mathbb{E}[\tau | \mathbf{L}] = \bar{\Delta}(\omega, \theta) + \varepsilon_{NT}(1) - \bar{\varepsilon}^0(\theta, \omega),   
$$ 
where $\bar{\Delta}(\omega, \theta) = \mathbb{E}[\hat{\tau} - \tau | \mathbf{L}]$. Typically, we may expect $\bar{\Delta}(\omega, \theta) \approx 0$ by the triple robustness property of our estimator (see Section \ref{sec:identification}). In addition, $\varepsilon^0(\omega, \theta) = O_p(\min\{||\theta||_2, ||\omega||_2 \})$, so that whenever $||\theta||_2, ||\omega||_2 = o(\sigma_{NT})$, the variance contribution of $ \varepsilon^0(\omega, \theta)$ is asymptotically negligible relative to the one from $\varepsilon_{NT}$. This occurs if the weights are sufficiently dispersed over the pre-treatment periods and units. Under such conditions 
 $$
 \hat{\tau} - \mathbb{E}[\tau|\mathbf{L}] = \varepsilon_{NT}(1) + o_p(\sigma_{NT}^2),   
$$  
and therefore the variance of $\hat{\tau}$ only depends (asymptotically) on the variance of $\varepsilon_{NT}(1)$. Here,  $\mathbb{E}[\tau | \mathbf{L}]$ defines the expected treatment effect on the treated unit, conditional on the latent factor $\mathbf{L}$ and on $T_0,N_0$.

 Existing literature has proposed multiple procedures to estimate such variance components. We provide a formal algorithm to conduct inference with multiple treated units in Algorithm \ref{alg:bootstrapVarEst} below.

\section{Extensions}  \label{sec:extensions}

\subsection{Estimation with arbitrary treated units and periods} \label{sec:inference2}

Next, we generalize this idea to multiple treated units and periods. 
As a first step, that will allow us to generalize the estimator to multiple treated units and periods, define 
$$
I_{j,s}^{i,t} = 1\{(i,t) = (j,s)\}
$$ 
the indicator of whether observation $(i,t)$ is the observation $(j,s)$. We will consider unit specific and time specific weights $\omega(\lambda), \theta(\lambda)$ parametrized by an arbitrary vector of parameters $\lambda$. Such parametrization will simplify the choice of the weights as we discuss further below. 

For each unit $(i,t)$ we construct an estimator  
\begin{equation} \label{eqn:auxiliary_tau2}
\small 
\begin{aligned} 
& \hat{\tau}_{it}(\lambda) =    \\ \\ &  \mathrm{arg} \min_{\tau_{i,t}}  \min_{\alpha,\beta,\mathbf{L}} \sum_{j=1}^{N}\sum_{s=1}^{T} \Big[(1 - W_{j,s}) + W_{j,s} I_{j,s}^{i,t}\Big] \omega_{j}^{i,t}(\lambda) \theta_s^{i,t}(\lambda) \left(Y_{js}-\alpha_{j}-\beta_{s}-L_{js} - \tau_{i,t} \cdot I_{j,s}^{i,t} \right)^{2}+\lambda_{nn}\left\Vert \mathbf{L}\right\Vert 
\end{aligned} 
\end{equation}
This corresponds to the estimated treatment effects on a given unit $(i,t)$, \textit{as if} unit $(i,t)$ was the only treated units, after removing the outcomes for the other treated observations over the treatment periods. We can write the final average treatment effect on the treated as 
$$
\hat{\tau} = \frac{1}{\sum_{i=1}^N \sum_{t=1}^T W_{i,t}} \sum_{i=1}^N \sum_{t=1}^T W_{i,t} \hat{\tau}_{i,t}(\lambda).  
$$ 

Such estimator generalizes the estimator we propose in the presence of a single treated unit and period. We summarize the estimation procedure in Algorithm \ref{alg:alg2}.

\begin{algorithm} [!ht]   \caption{Triple robust panel estimator for treatment effect with multiple treated units}\label{alg:alg2}
    \begin{algorithmic}[1]

   \Require Grid of values $\mathcal{G}$ for $(\lambda_{time},\lambda_{unit},\lambda_{nn})$, treatments $\mathbf{W}$, outcomes $\mathbf{Y}$   
   \ForEach{$\Lambda \in \mathcal{G}$}
   \State For each $(i,t)$ such that $W_{it} = 0$ estimate $\hat{\tau}_{it}(\Lambda)$ as in Equation \eqref{eqn:auxiliary_tau2}
   \State Compute the function $Q(\Lambda)$ in Equation \eqref{eqn:q_lambda}
   $$
   Q(\lambda)=\sum_{i=1}^{N}\sum_{t=1}^{T}(1-W_{it}) \Big(\hat{\tau}_{it}(\lambda)\Big)^2 
   $$ 
 \EndFor 
 \State Find $\hat{\lambda} \in \mathrm{arg} \min_{\lambda \in \mathcal{G}} Q(\lambda)$
\State Compute the estimator $$
\hat{\tau} = \frac{1}{\sum_{i=1}^N \sum_{t=1}^T W_{i,t}} \sum_{i=1}^N \sum_{t=1}^T W_{i,t} \hat{\tau}_{i,t}(\hat{\lambda}).  
$$ 
where $\hat{\tau}_{i,t}(\lambda)$ are defined as in Equation \eqref{eqn:auxiliary_tau2}. 
         \end{algorithmic}
\end{algorithm}

To conduct inference, we propose a simple non-parametric bootstrap estimator, that is formally discussed in Algorithm \ref{alg:bootstrapVarEst}. The bootstrap does not impose assumptions on time independence of idiosyncratic shocks within each unit, but its validity requires a growing number of treated units.

\begin{algorithm}[!ht] 
\caption{Bootstrap Variance Estimation with multiple treated units}
\label{alg:bootstrapVarEst}
\begin{algorithmic}[1]
\State \textbf{Data:} $Y, W, B$
\State \textbf{Result:} Variance estimator $\widehat{V}^{\mathrm{cb}}_{\tau}$
\For{$b \gets 1$ to $B$}
    \State Construct a bootstrap dataset $\bigl(Y^{(b)}, W^{(b)}\bigr)$ by sampling $N_0$ rows of $(Y, W)$ with replacement for the first $N_0$ units and sampling $N_1$ rows with replacement of $(Y, W)$ from the last $N_1$ units. 
    \State Compute the TROP estimator $\hat{\tau}^{(b)}$ from $\bigl(Y^{(b)}, W^{(b)}\bigr)$.
\EndFor
\State Define
\[
\widehat{V}_{\tau} \;=\;
\frac{1}{B} \sum_{b=1}^B 
\bigl(
   \hat{\tau}^{(b)} - 
   \tfrac{1}{B} \sum_{b=1}^B \hat{\tau}^{(b)}
\bigr)^2\,.
\]
\end{algorithmic}
\end{algorithm}

 \begin{rem}[Alternative approaches to estimate the effect on multiple treated units] Note that here we estimate each individuals treatment effect $\hat{\tau}_{i,t}$  separately for each unit, hence allowing for heterogeneous effects across units and time periods. It is possible however to impose some additional modeling assumption to gain in terms of precision. For example, under homogeneous treatment effects, one could aggregate estimators for different units using different weighing mechanisms. We omit further details and leave a study of efficient estimation of treatment effects to future research. \qed  
 \end{rem}

\subsection{The TROP Estimator with Covariates}

In the presence of covariates, researchers may consider multiple extensions of the TROP estimator. One simple extension we propose is to include covariates additively into the model. That is, we can parametrize 
$$
L_{j,s} = X_{j,s} \phi + R_{j,s} 
$$ 
where $X_{j,s}$ defines observable covariates and $R_{j,s}$ denotes a low-rank component. The main advantage of controlling for covariates in this case is to relax the low rank assumption of $L_{j,s}$, assuming instead that $L_{j,s}$ is given by the sum of observed covariates and of a possibly low-rank component $R_{j,s}$. In this case the objective function for a treated unit/period $(i^\star, t^\star)$ reads as 

 \begin{equation}   
\hat{\tau}_{i^\star, t^\star} = \mathrm{arg} \min_\tau \min_{\alpha,\beta, \phi, \mathbf{R} }\sum_{j=1}^{N}\sum_{s=1}^{T}  \theta_s^{i^\star, t^\star} \omega_j^{i^\star, t^\star} \left(Y_{js}-\alpha_{j}-\beta_{s}-X_{j,s} \phi - R_{j,s} - \tau W_{j,s} \right)^{2}+\lambda_{nn}\left\Vert \mathbf{R}\right\Vert,\label{eq:doubly-weighted-regression_covariates}
\end{equation} 

We can then follow verbatim our procedure also controlling for covariates in this case. Appendix Theorem \ref{cor:covariates-bias} extends our triple-robustness results in the presence of covariates without assuming a particular estimator for $\phi$. Note, importantly, that with covariates the choice of the optimal weights may also change since the objective function in Algorithm \ref{alg:alg1} now optimizes over the tuning parameters after residualizing covariates from the outcomes (similar, for instance to the SDID method).

\section{Conclusion} \label{sec:conclusions}

This paper develops the Triply RObust Panel (TROP) estimator—a unified framework that combines data-driven unit weights, time-decay weights, and a regression adjustment to estimate counterfactuals in panel-data settings with potentially complex assignment patterns. The estimator inherits a triple robustness property, as its bias depends on the product of unit imbalance, time imbalance, and misspecification of the regression adjustment. We propose a leave-one-out cross-validation strategy that simultaneously tunes the unit and time weights, as well as the strength of the nuclear-norm penalty for estimating a low-rank factor model. Extensive simulations calibrated to canonical CPS and PWT-World-Table, as well as other applications, illustrate the benefits of the methods.  

These results highlight two broader lessons for empirical work with panels. First, the recent past typically matters most, and estimators should adaptively adjust to learning the most relevant periods without considering observations over time as exchangeable. Second, standard DID methods may approximate poorly relevant economic applications, and the combination of balancing and interactive fixed effects can provide more accunemp-rate counterfactuals.

Several open questions remain. First, alternative distance metrics (e.g., covariate-based Mahalanobis distances) could further refine unit weighting. Second, integrating TROP with augmented inverse-propensity weighting or proxy-variable strategies promises additional gains when assignment probabilities are partially observed. Finally, we leave the study of such estimators in contexts with dynamic treatment effects to future research.

\bibliography{\bib}

\newpage

%\citep{abadie2003, abadie2010synthetic, abadie2014} 

\section{Proofs} \label{proof:thm:1}

\newcommand{\EE}{\mathbb{E}}
\newcommand{\inner}[2]{\left\langle #1,\, #2 \right\rangle}
\newcommand{\opnorm}[1]{\left\| #1 \right\|_{\mathrm{op}}}
\newcommand{\fnorm}[1]{\left\| #1 \right\|_{F}}
\newcommand{\nnorm}[1]{\left\| #1 \right\|_{*}}

\subsection{General notation} 

Recall the general setup: 
$Y_{it}(0)=L_{it}+\varepsilon_{it}$ with $\EE[\varepsilon_{it}\mid L]=0$ for all $(i,t)$.
There is a single treated unit/period $(N,T)$; all other cells are pre-treatment.
Let $\omega\in\mathbb{R}^{N_0}$ and $\theta\in\mathbb{R}^{T_0}$ be fixed weights that each sum to one:
$\sum_{i\le N_0}\omega_i=\sum_{t\le T_0}\theta_t=1$, where $T_0 = T - 1$, $N_0 = N - 1$. (Note that our results extend beyond one treated unit/period, at the expense of more tedious notation). 

We introduce the following notation: 
\[
(U_\omega)_i:=\begin{cases}\omega_i,& i\le N_0\\ 0,& i>N_0,\end{cases}
\qquad
(V_\theta)_t:=\begin{cases}\theta_t,& t\le T_0\\ 0,& t>T_0,\end{cases}
\]
and the canonical basis vectors $e_N\in\mathbb{R}^N$, $e_T\in\mathbb{R}^T$ so that $e_N$ has the $N^{th}$ entry equal to one and all other entries equal to zero and $e_T$ has the $T^{th}$ entry equal to one and all other entries equal to zero.
Define 
\[
M(\theta,\omega)\;:=\;(e_N-U_\omega)\,(V_\theta-e_T)^\top\in\mathbb{R}^{N\times T}.
\]
We use the Frobenius inner product $\inner{A}{B}=\sum_{i,t}A_{it}B_{it}$.

Let $\widehat{\mathbf{L}}=\{\widehat{\mathbf{L}}_{it}\}$ be any (possibly data-dependent) adjustment.
Define the TROP counterfactual for $(N,T)$ by
\begin{equation}\label{eq:TROP-four-cells}
\hat{\mu}^{\mathrm{TROP}}(0;\theta,\omega)
:=\widehat{\mathbf{L}}_{NT}
+\sum_{t\le T_0}\theta_t\,(Y_{Nt}-\widehat{\mathbf{L}}_{Nt})
+\sum_{i\le N_0}\omega_i\,(Y_{iT}-\widehat{\mathbf{L}}_{iT})
-\sum_{i\le N_0}\sum_{t\le T_0}\omega_i\theta_t\,(Y_{it}-\widehat{\mathbf{L}}_{it}).
\end{equation}

\subsection{Main theorem}

The following lemma provides a useful representation of the TROP estimator. 

\begin{lemma} \label{lem:rankone}
With $M(\theta,\omega)$ as above,
\begin{equation}\label{eq:R}
\hat{\mu}^{\mathrm{TROP}}(0;\theta,\omega)
=Y_{NT}+\inner{\,Y-\widehat{\mathbf{L}}\,}{\,M(\theta,\omega)\,}.
\end{equation}
\end{lemma}

\begin{proof}[Proof of Lemma~\ref{lem:rankone}]
Expand the right-hand side of \eqref{eq:R}. Since $M_{Nt}=\theta_t$ for $t\le T_0$,
$M_{iT}=\omega_i$ for $i\le N_0$, $M_{it}=-\omega_i\theta_t$ for $i\le N_0,\ t\le T_0$,
and $M_{NT}=-1$, we obtain
\[
\inner{Y-\widehat{\mathbf{L}}}{M}
=\sum_{t\le T_0}\theta_t\,(Y_{Nt}-\widehat{\mathbf{L}}_{Nt})
+\sum_{i\le N_0}\omega_i\,(Y_{iT}-\widehat{\mathbf{L}}_{iT})
-\sum_{i\le N_0}\sum_{t\le T_0}\omega_i\theta_t\,(Y_{it}-\widehat{\mathbf{L}}_{it})
-\,(Y_{NT}-\widehat{\mathbf{L}}_{NT}).
\]
Adding $Y_{NT}$ cancels the last term and leaves exactly \eqref{eq:TROP-four-cells}.
\end{proof}

\begin{theorem}\label{thm:gen-identity}
Under Assumption \ref{ass:factor_model} and for fixed $(\theta,\omega)$,
\begin{align}
\hat{\mu}^{\mathrm{TROP}}(0;\theta,\omega)- Y_{NT}(0)
&=\ \inner{\,\mathbf{L}-\widehat{\mathbf{L}}\,}{\,M(\theta,\omega)\,}
\ + \bar\varepsilon_0(\theta,\omega) - \varepsilon_{NT}, 
\label{eq:gen-identity}
\end{align}
with $\bar\varepsilon_0(\theta,\omega)$ defined in Equation \eqref{eqn:epsilon_zero}. 
In particular, taking $\EE[\cdot\mid \mathbf{L}]$,
\begin{equation}\label{eq:cond-bias}
\EE\!\left[\hat{\mu}^{\mathrm{TROP}}(0;\theta,\omega)-Y_{NT}(0)\ \bigm|\ \mathbf{L}\right]
=\ \inner{\,\mathbf{L}-\EE[\widehat{\mathbf{L}}\mid \mathbf{L}]\,}{\,M(\theta,\omega)\,}.
\end{equation}
\end{theorem}

\begin{proof}[Proof of Theorem~\ref{thm:gen-identity}]
For short, we will refer to $M(\theta,\omega)$ as $M$. 
From Lemma~\ref{lem:rankone} and Assumption \ref{ass:factor_model}, by letting $Y$ denote the matrix with $(it)^{th}$ entry  $Y_{it}$, 
\begin{equation} \label{eqn:jk0}
\hat{\mu}^{\mathrm{TROP}}(0;\theta,\omega)- Y_{NT}(0)
=\inner{\,Y-\widehat{\mathbf{L}}\,}{\,M\,}+(Y_{NT}-\mathbf{L}_{NT}) - \varepsilon_{NT},
\end{equation} 
where we decompose $Y_{NT}(0) = \mathbf{L}_{NT} + \varepsilon_{NT}$. In addition, we can write 
\begin{equation} \label{eqn:jk}
\inner{\,Y-\widehat{\mathbf{L}}\,}{\,M\,} = \inner{\,\mathbf{L}-\widehat{\mathbf{L}}\,}{\,M\,} + \inner{\,\varepsilon\,}{\,M\,} - \Big(Y_{NT} - \mathbf{L}_{NT} - \varepsilon_{NT}\Big) 
\end{equation} 
because $M_{NT}=-1$, the contribution of $(N,T)$ inside $\inner{\mathbf{L} + \varepsilon}{M}$ equals $-(\mathbf{L}_{NT} + \varepsilon_{NT})$ and the contribution of $(N,T)$ inside $\inner{Y}{M}$ equals $-Y_{NT}$.  

Combining Equations \eqref{eqn:jk0} and \eqref{eqn:jk}, it follows that 
$$
\hat{\mu}^{\mathrm{TROP}}(0;\theta,\omega)- Y_{NT}(0) = \inner{\,\mathbf{L}-\widehat{\mathbf{L}}\,}{\,M\,} + \inner{\,\varepsilon\,}{\,M\,}
$$
which equals the desired result in Equation \eqref{eq:gen-identity} by definition of $M$. 

Equation \eqref{eq:cond-bias} follows immediately after taking expectations, since $\mathbb{E}[\varepsilon_{it} | \mathbf{L}] = 0$ for all $(i,t)$ and $\theta, \omega$ are fixed (not data dependent). 
\end{proof}

\subsection{Proof of Theorem \ref{thm:1}}

\begin{corollary} \label{cor:main}
 Let Assumptions \ref{ass:factor_model} and \ref{ass:class_estimators} hold.
Then 
\begin{equation}\label{eq:versionB}
\begin{aligned} 
\EE\!\left[\hat{\mu}^{\mathrm{TROP}}(0;\theta,\omega)-Y_{NT}(0)\ \bigm|\ \mathbf{L}\right]
& =\ \,\Delta_u(\omega,\Gamma)^\top\,B\,\Delta_t(\theta,\Lambda) \\
&\le ||\Delta_u(\omega,\Gamma)||_2 ||\Delta_t(\theta,\Lambda)||_2 ||B||_{\star}.
\end{aligned} 
\end{equation}
\end{corollary}

\begin{proof}
Theorem \ref{thm:gen-identity} gives the identity
\begin{equation}\label{eq:cond-bias-8p}
\EE\!\left[\hat{\mu}^{\mathrm{TROP}}(0;\theta,\omega)-Y_{NT}(0)\ \bigm|\ \mathbf{L}\right] =\ \inner{\,\mathbf{L}-\EE[\widehat{\mathbf{L}}\mid \mathbf{L}]\,}{\,M(\theta,\omega)\,},
\end{equation}
where $M(\theta,\omega)=(e_N-U_\omega)(V_\theta-e_T)^\top$.

Under Assumptions \ref{ass:factor_model} and \ref{ass:class_estimators}, $\mathbf{L}-\EE[\widehat{\mathbf{L}}\mid \mathbf{L}]=-\Gamma B \Lambda^\top$.
Let $a:=e_N-U_\omega$ and $b:=V_\theta-e_T$. By Lemma \ref{lem:rankone-ip}, 
\begin{equation}\label{eq:rank-one-ip}
\inner{\,\Gamma B \Lambda^\top\,}{\,ab^\top\,}
= a^\top\Gamma B \Lambda^\top b
= (\Gamma^\top a)^\top B\,(\Lambda^\top b).
\end{equation}
Finally, note that by definition of $a=e_N-U_\omega$,
\[
\Gamma^\top a
=\Gamma^\top e_N - \Gamma^\top U_\omega
=\Gamma_N - \sum_{i\le N_0}\omega_i\Gamma_i
= -\big(\bar\Gamma_0(\omega)-\Gamma_N\big)
= -\,\Delta_u(\omega,\Gamma).
\]
Similarly, with $b=V_\theta-e_T$,
\[
\Lambda^\top b
=\Lambda^\top V_\theta - \Lambda^\top e_T
=\sum_{t\le T_0}\theta_t\Lambda_t - \Lambda_T
=\bar\Lambda_0(\theta)-\Lambda_T
= \Delta_t(\theta,\Lambda).
\]
Substituting these two identities into \eqref{eq:rank-one-ip} yields
\[
-\inner{\,\Gamma B\Lambda^\top\,}{\,ab^\top\,}
= \Delta_u(\omega,\Gamma)^\top B\,\Delta_t(\theta,\Lambda),
\]
proving Equation \eqref{eq:versionB}.

Finally, the last inequality in Equation \eqref{eq:versionB} follows directly from Lemma \ref{lem:nuclear-operator}. 
\end{proof}

Theorem \ref{thm:1} follows directly from Corollary \ref{cor:main} since $\mathbb{E}[\hat{\tau} - \tau | \mathbf{L}] = \mathbb{E}[\hat{\mu}^{\mathrm{TROP}}(0;\theta, \omega) - Y_{NT}(0)|\mathbf{L}]$.

\subsection{Triple robustness with covariates} \label{sec:estimator_covariates}

In the following theorem we generalize the triple robustness properties in the presence of covariates, allowing general (possibly penalized) regression methods in the presence of covariates. 

\begin{ass}[Factor model with covariates] \label{ass:factor_with_covariates} Suppose that for all $(i,t)$, $Y_{i,t}(0) = \mathbf{R}_{i,t} + X_{i,t} \beta + \varepsilon_{i,t}$, with $\mathbb{E}[\varepsilon_{i,t} | \mathbf{R}, \mathbf{X}] = 0$ and $\mathbf{R}_{i,t} = \Gamma \Lambda^\top$. 
\end{ass} 

Assumption \ref{ass:factor_with_covariates} generalizes Assumption \ref{ass:factor_model} introducing additional observed covariates $X_{i,t}$.

Define the covariate contrast vector
\begin{equation}\label{eq:g-def}
g(\theta,\omega)
\;:=\;\sum_{i,t} M(\theta,\omega)_{it}\,\mathbf{X}_{it}
\;\;=\;\;(e_N-U_\omega)^\top \mathbf{X} (V_\theta-e_T)\ \in\ \mathbb{R}^p,
\end{equation}
and the estimator's bias 
\[
\delta_\phi \;:=\; \EE[\widehat\phi \mid \mathbf{X},\mathbf{R}]-\phi \in \mathbb{R}^p.
\]

\begin{ass}[Estimators with covariates] \label{ass:estimator_covariates} There exists a matrix $B\in\mathbb{R}^{K\times K}$ such that
\[
\EE\!\left[\,\widehat{\mathbf{L}} - \mathbf{L} \ \bigm|\ \mathbf{X},\mathbf{R}\,\right] \;=\; \mathbf{X}\delta_\phi \;+\; \Gamma B \Lambda^\top \;+\; E_{\perp}, 
\]
for an arbitrary matrix $E_{\perp} \in \mathbb{R}^{T \times N}$, with $||E||_{\max} = o(\eta)$ for some small $\eta \ge 0$. 
\end{ass} 

Assumption \ref{ass:estimator_covariates} mimics Assumption \ref{ass:class_estimators} introducing additional covariates. Note that we impose no restrictions on the estimation of $\hat{\beta}$. Here, estimation of $\hat{\beta}$ interacts with properties of the estimator $\hat{\mathbf{R}}$. Therefore, in addition to the matrix $B$ and bias component $\delta_\phi$ that capture misspecification of the low-rank regression adjustment and of the estimation of the coefficients, we also allow one additional term that may violate the invariance condition in Assumption \ref{ass:class_estimators}. We think of $E_{\perp}$ as a matrix with element-wise components of small order capturing any additional approximation error.

\begin{thm}[Bias decomposition with covariates]\label{cor:covariates-bias}
Let Assumptions \ref{ass:factor_with_covariates} and \ref{ass:estimator_covariates} hold. Suppose that $\omega, \theta \ge 0$. Then 
\begin{equation}\label{eq:cov-bias-decomp}
\EE\!\left[\hat{\mu}^{\mathrm{TROP}}(0;\theta,\omega)-Y_{NT}(0)\ \bigm|\ \mathbf{X}, \mathbf{R} \right]
\;=\; \,\Delta_u(\omega,\gamma)^\top\,B\,\Delta_t(\theta,\lambda)
\;-\; g(\theta,\omega)^\top \delta_\phi + o(\eta).
\end{equation}
Consequently,
\begin{align}
\Big|\EE[\hat{\mu}^{\mathrm{TROP}}(0;\theta,\omega)-Y_{NT}(0)\mid \mathbf{X}, \mathbf{R} ]\Big|
&\le \ \|\Delta_u(\omega,\gamma)\|_2\ \|B\|_*\ \|\Delta_t(\theta,\lambda)\|_2 \notag\\
&\hspace{1.9em}+\ \|g(\theta,\omega)\|_2\ \|\delta_\phi\|_2 + o(\eta). \label{eq:cov-bound}
\end{align}
\end{thm}

\begin{proof}
We apply Theorem \ref{thm:gen-identity} with $\mathbf{L}:=\mathbf{X}\beta+\mathbf{R}$ and adjustment $\widehat{\mathbf{L}}=\mathbf{X}\widehat{\beta}+\widehat{\mathbf{R}}$:
\[
\EE\!\left[\widehat\tau_{\mathrm{TROP}}-Y_{NT}(0)\ \bigm|\ \mathbf{X},\mathbf{R}\right]
=\ \inner{\,\mathbf{L} -\EE[\widehat{\mathbf{L}}\mid \mathbf{X},\mathbf{R}]\,}{\,M(\theta,\omega)\,}.
\]
Taking inner products with $M(\theta,\omega)$ and using yields
\[
\EE[\widehat\tau_{\mathrm{TROP}}-Y_{NT}(0) \mid \mathbf{X},\mathbf{R}]
=\ -\,\inner{X\delta_\phi}{M(\theta,\omega)}\ -\ 
\inner{\Gamma B \Lambda^\top}{(e_N-U_\omega)(V_\theta-e_T)^\top} - \inner{E_{\perp}}{M(\theta, \omega)}.
\]
For the covariate term, by Lemma \ref{lem:rankone}
$\inner{X\delta_\phi}{ab^\top}=a^\top(X\delta_\phi)b$ with
$a=e_N-U_\omega$, $b=V_\theta-e_T$, we get
\[
\inner{X\delta_\phi}{M(\theta,\omega)}=(e_N-U_\omega)^\top X (V_\theta-e_T)\,\delta_\phi
= g(\theta,\omega)^\top \delta_\phi.
\]
For the residual term, use the same identity with $A=I_K-B$ to obtain
\[
-\inner{\Gamma B \Lambda^\top}{ab^\top}
=- \big(\Gamma^\top a\big)^\top B \,\big(\Lambda^\top b\big)
= \,\Delta_u(\omega,\gamma)^\top B\,\Delta_t(\theta,\lambda),
\]
because $\Gamma^\top a=\Gamma_N-\bar\Gamma_0(\omega)=-\Delta_u$ and
$\Lambda^\top b=\bar\Lambda_0(\theta)-\Lambda_T=\Delta_t$.
Finally, it is easy to show that by Holder's inequality 
$$
|\inner{E_{\perp}}{M} | \le |e_N^\top E V_\theta| + |e_N^\top E e_T| + |U_\omega^\top E V_\theta| + |U_\omega^\top E e_T| \le ||E||_{\max}(1 + ||\omega||_1)(1 + ||\theta||_1). 
$$

Combining the displays gives \eqref{eq:cov-bias-decomp}, since $||\omega||_1 = ||\theta||_1 = 1$ by assumption in the theorem. 

For the bound \eqref{eq:cov-bound}, we apply Lemma \ref{lem:nuclear-operator}:
\[
\big|\Delta_u^\top B\Delta_t\big|
=\big|\inner{B}{\Delta_u\Delta_t^\top}\big|
\le \|B\|_*\ \|\Delta_u\Delta_t^\top\|_{\mathrm{op}}
=\|B\|_*\ \|\Delta_u\|_2\ \|\Delta_t\|_2,
\]
and Cauchy–Schwarz to the covariate term $|g^\top\delta_\phi|\le \|g\|_2\|\delta_\phi\|_2$.
\end{proof}

\subsection{Auxiliary lemmas}

\begin{lemma} \label{lem:rankone-ip}
For any $X\in\mathbb{R}^{N\times T}$ and vectors $a\in\mathbb{R}^N$, $b\in\mathbb{R}^T$. Then 
\[
\langle X,\;ab^\top\rangle \;=\; a^\top X b.
\]
\end{lemma}

\begin{proof}
By definition of the Frobenius inner product,
\[
\langle X,ab^\top\rangle \;=\; \mathrm{tr}\!\big(X^\top ab^\top\big)
\;=\; \mathrm{tr}\!\big(b^\top X^\top a\big)
\;=\; a^\top X b.
\]
\end{proof}

\begin{lemma}\label{lem:nuclear-operator}
For any $A\in\mathbb{R}^{K\times K}$ and vectors $x,y\in\mathbb{R}^K$,
\[
\big|\,x^\top A y\,\big|
\;=\;\big|\inner{A}{xy^\top}\big|
\;\le\; \nnorm{A}\;\opnorm{xy^\top}
\;\le\; \nnorm{A}\;\|x\|_2\,\|y\|_2.
\]
\end{lemma}

\begin{proof}
The equality $x^\top A y=\inner{A}{xy^\top}$ follows from Lemma \ref{lem:rankone-ip}. Hölder’s inequality for the dual pair
$(\|\cdot\|_*,\|\cdot\|_{\mathrm{op}})$, where $||\cdot||_{\mathrm{op}}$ denotes the operator norm yields
$\big|\inner{A}{Z}\big|\le \nnorm{A}\,\opnorm{Z}$ for all $A,Z$ of matching size.
Take $Z=xy^\top$ to get
$\big|\inner{A}{xy^\top}\big|\le \nnorm{A}\,\opnorm{xy^\top}$.
Finally, by definition,
\[
\|xy^\top\|_{\mathrm{op}}
= \sup_{\|v\|_2=1}\,\|\,xy^\top v\,\|_2
= \sup_{\|v\|_2=1}\,\|\,x\, (y^\top v)\,\|_2
\le \|x\|_2 \cdot \sup_{\|v\|_2=1} |y^\top v|.
\]
By Cauchy–Schwarz, $\sup_{\|v\|_2=1} |y^\top v|=\|y\|_2$, with equality at $v=y/\|y\|_2$ (or any unit vector aligned with $y$). Therefore
$
\|xy^\top\|_{\mathrm{op}} \le \|x\|_2\,\|y\|_2.
$
\end{proof}

\newpage 

\section{Additional Simulation Results}

\begin{table}[!ht]
\captionsetup{font=scriptsize}
\centering
\scalebox{0.6}{
\begin{tabular}{|c|c|c|c|c|c|c|c|c|c|c|c|c|c|c|}
\hline
\multicolumn{15}{|l|}{\textbf{Panel A:} $T_{\mathrm{post}}=1$ (treated periods), $N_{\mathrm{tr}}=10$ (treated units)}\\
\hline
\multirow{2}{*}{Dataset} & \multirow{2}{*}{Treatment} & \multirow{2}{*}{$N$} & \multirow{2}{*}{$T$} &
\multirow{2}{*}{$\frac{\|F\|_{F}}{\sqrt{NT}}$} & \multirow{2}{*}{$\frac{\|M\|_{F}}{\sqrt{NT}}$} &
\multirow{2}{*}{$\sqrt{\frac{Tr(\Sigma)}{T}}$} & \multirow{2}{*}{AR(2)} &
\multirow{2}{*}{$(\lambda_{unit},\lambda_{time},\lambda_{nn})$} &
\multicolumn{6}{c|}{RMSE} \\
\cline{10-15}
 & & & & & & & & &
TROP & SDID & SC & DID & MC & DIFP \\
\hline
CPS logwage & min wage & 50 & 40 & 0.992 & 0.10  & 0.098 & $(0.010,-0.057)$ & $(0, 0.3, 0.0178)$  & \textbf{0.049} & 0.053 & 0.059 & 0.067 & 0.052 & 0.054 \\
\hline
CPS urate   & min wage & 50 & 40 & 0.770 & 0.40  & 0.604 & $(0.039,0.011)$  & $(0, 0.2, 0.0031)$  & 0.290          & 0.301 & 0.302 & \textbf{0.267} & 0.288 & 0.300 \\
\hline
PWT         & democracy&111 & 48 & 0.972 & 0.229 & 0.069 & $(0.913,-0.221)$ & $(0.3,3.4,0.01)$     & \textbf{0.017} & 0.020 & 0.024 & 0.167 & 0.020 & 0.023 \\
\hline
Germany     & random   & 17 & 44 & 0.994 & 0.108 & 0.032 & $(0.840,-0.200)$ & $(0,3.2,0.01)$       & \textbf{0.018} & 0.021 & 0.050 & 0.068 & 0.023 & 0.033 \\
\hline
Basque      & random   & 18 & 43 & 0.986 & 0.167 & 0.037 & $(1.228,-0.577)$ & $(0.3,0.2,0.005)$    & \textbf{0.020} & 0.023 & 0.047 & 0.129 & 0.027 & 0.041 \\
\hline
Smoking     & random   & 39 & 31 & 0.935 & 0.337 & 0.166 & $(0.506,-0.058)$ & $(0.8, 0.3, 0.0316)$ & \textbf{0.066} & 0.068 & 0.077 & 0.169 & 0.070 & 0.076 \\
\hline
Boatlift    & random   & 44 & 19 & 0.913 & 0.344 & 0.276 & $(0.031,-0.087)$ & $(0.3, 0.4, 0.2014)$ & \textbf{0.114} & 0.126 & 0.144 & 0.174 & 0.123 & 0.145 \\
\hline

\multicolumn{15}{|l|}{\textbf{Panel B:} $T_{\mathrm{post}}=10$ (treated periods), $N_{\mathrm{tr}}=1$ (treated unit)}\\
\hline
\multirow{2}{*}{Dataset} & \multirow{2}{*}{Treatment} & \multirow{2}{*}{$N$} & \multirow{2}{*}{$T$} &
\multirow{2}{*}{$\frac{\|F\|_{F}}{\sqrt{NT}}$} & \multirow{2}{*}{$\frac{\|M\|_{F}}{\sqrt{NT}}$} &
\multirow{2}{*}{$\sqrt{\frac{Tr(\Sigma)}{T}}$} & \multirow{2}{*}{AR(2)} &
\multirow{2}{*}{$(\lambda_{unit},\lambda_{time},\lambda_{nn})$} &
\multicolumn{6}{c|}{RMSE} \\
\cline{10-15}
 & & & & & & & & &
TROP & SDID & SC & DID & MC & DIFP \\
\hline
CPS logwage & min wage & 50 & 40 & 0.992 & 0.10  & 0.098 & $(0.010,-0.057)$ & $(0.6, 0.3, 0.009)$   & \textbf{0.061} & 0.069 & 0.074 & 0.126 & 0.078 & 0.085 \\
\hline
CPS urate   & min wage & 50 & 40 & 0.770 & 0.40  & 0.604 & $(0.039,0.011)$  & $(0.7, 0, 0.005)$     & \textbf{0.370} & 0.391 & 0.400 & 0.517 & 0.372 & 0.382 \\
\hline
PWT         & democracy&111 & 48 & 0.972 & 0.229 & 0.069 & $(0.913,-0.221)$ & $(0.3,0.4,0.005)$     & \textbf{0.070} & 0.106 & 0.134 & 0.318 & 0.095 & 0.116 \\
\hline
Germany     & random   & 17 & 44 & 0.994 & 0.108 & 0.032 & $(0.840,-0.200)$ & $(0.6,0.2,0.015)$     & \textbf{0.054} & 0.080 & 0.103 & 0.148 & 0.073 & 0.095 \\
\hline
Basque      & random   & 18 & 43 & 0.986 & 0.167 & 0.037 & $(1.228,-0.577)$ & $(0.6,0.3,0.01)$      & \textbf{0.088} & 0.139 & 0.235 & 0.329 & 0.139 & 0.139 \\
\hline
Smoking     & random   & 39 & 31 & 0.935 & 0.337 & 0.166 & $(0.506,-0.058)$ & $(0.4, 0.3, 0.0265)$  & \textbf{0.230} & 0.283 & 0.370 & 0.507 & 0.279 & 0.315 \\
\hline
Boatlift    & random   & 44 & 19 & 0.913 & 0.344 & 0.276 & $(0.031,-0.087)$ & $(0.1, 0.2, 0.255)$   & \textbf{0.329} & 0.440 & 0.480 & 0.438 & 0.346 & 0.461 \\
\hline
\end{tabular}
}
\caption{Simulation design as in Table \ref{tab:merged_panels_Tpost_Ntr}, varying the number of post-treatment periods and treated units. \textbf{Panel A} uses one treated period and ten treated units ($T_{\mathrm{post}}{=}1$, $N_{\mathrm{tr}}{=}10$). \textbf{Panel B} uses ten treated periods and one treated unit ($T_{\mathrm{post}}{=}10$, $N_{\mathrm{tr}}{=}1$).}
\label{tab:two_panels_Tpost_Ntr}
\end{table}

\begin{table}[!ht]
\captionsetup{font=scriptsize}
\centering
\scalebox{0.6}{
\begin{tabular}{|l|
                S[table-format=1.3]|
                S[table-format=1.3]|
                S[table-format=1.3]|
                l|
                l|
                S[table-format=+1.4]|S[table-format=+1.4]|S[table-format=+1.4]|S[table-format=+1.4]|S[table-format=+1.4]|S[table-format=+1.4]|
               }
\hline
\multicolumn{12}{|l|}{\textbf{Panel A: PWT data (Y: log-GDP)}}\\
\hline
\multicolumn{1}{|c|}{Treatment} &
\multicolumn{1}{c|}{$\frac{\|F\|_{F}}{\sqrt{NT}}$} &
\multicolumn{1}{c|}{$\frac{\|M\|_{F}}{\sqrt{NT}}$} &
\multicolumn{1}{c|}{$\sqrt{\frac{Tr(\Sigma)}{T}}$} &
\multicolumn{1}{c|}{AR(2)} &
\multicolumn{1}{c|}{$(\lambda_{unit},\lambda_{time},\lambda_{nn})$} &
\multicolumn{6}{c|}{bias} \\
\cline{7-12}
\multicolumn{6}{|c|}{} &
\multicolumn{1}{c|}{TROP} & \multicolumn{1}{c|}{SDID} & \multicolumn{1}{c|}{SC} &
\multicolumn{1}{c|}{DID}  & \multicolumn{1}{c|}{MC}   & \multicolumn{1}{c|}{DIFP} \\
\hline
Democracy & 0.972 & 0.229 & 0.069 & $(0.913,-0.221)$ & $(0.3,0.4,0.006)$ & 0.010 & -0.011 & 0.001 & 0.176 & 0.035 & -0.005 \\
\hline
No AR     & 0.972 & 0.229 & 0.069 & $(0,0)$          & $(0.2,0.3,0.016)$   & 0.011 & 0.010  & 0.001  & 0.181 & 0.041 & -0.008 \\
\hline
Education & 0.972 & 0.229 & 0.069 & $(0.913,-0.221)$ & $(0.75,0.275,0.026)$& 0.009 & -0.003 & 0.028  & 0.168 & 0.040 & -0.007 \\
\hline
Random    & 0.972 & 0.229 & 0.069 & $(0.913,-0.221)$ & $(0.4,0.45,0.003)$  & 0.001 & -0.002 & 0.0004 & 0.002 & -0.004 & 0.0006 \\
\hline

\multicolumn{12}{|l|}{\textbf{Panel B: CPS data (Y: log-wage)}}\\
\hline
\multicolumn{1}{|c|}{Treatment} &
\multicolumn{1}{c|}{$\frac{\|F\|_{F}}{\sqrt{NT}}$} &
\multicolumn{1}{c|}{$\frac{\|M\|_{F}}{\sqrt{NT}}$} &
\multicolumn{1}{c|}{$\sqrt{\frac{Tr(\Sigma)}{T}}$} &
\multicolumn{1}{c|}{AR(2)} &
\multicolumn{1}{c|}{$(\lambda_{unit},\lambda_{time},\lambda_{nn})$} &
\multicolumn{6}{c|}{bias} \\
\cline{7-12}
\multicolumn{6}{|c|}{} &
\multicolumn{1}{c|}{TROP} & \multicolumn{1}{c|}{SDID} & \multicolumn{1}{c|}{SC} &
\multicolumn{1}{c|}{DID}  & \multicolumn{1}{c|}{MC}   & \multicolumn{1}{c|}{DIFP} \\
\hline
Minimum wage & 0.992 & 0.10  & 0.098 & $(0.010,-0.057)$ & $(0,0.1,0.9)$ & 0.008 & 0.006 & 0.019 & 0.022 & 0.012 & 0.006 \\
\hline
Gun Law      & 0.992 & 0.10  & 0.098 & $(0.010,-0.057)$ & $(0,0.35,0.041)$   & -0.0002 & 0.008 & -0.004 & 0.012 & 0.015 & 0.009 \\
\hline
Abortion     & 0.992 & 0.10  & 0.098 & $(0.010,-0.057)$ & $(0,0.2,0.281)$    & 0.004 & 0.004 & 0.016 & 0.003 & 0.003 & 0.001 \\
\hline
Random       & 0.992 & 0.10  & 0.098 & $(0.010,-0.057)$ & $(0,0.2,0.21)$     & -0.0003 & 0.001 & -0.001 & -0.0007 & 0.001 & -0.0006 \\
\hline

\multicolumn{12}{|l|}{\textbf{Panel C: CPS data (Y: U-rate)}}\\
\hline
\multicolumn{1}{|c|}{Treatment} &
\multicolumn{1}{c|}{$\frac{\|F\|_{F}}{\sqrt{NT}}$} &
\multicolumn{1}{c|}{$\frac{\|M\|_{F}}{\sqrt{NT}}$} &
\multicolumn{1}{c|}{$\sqrt{\frac{Tr(\Sigma)}{T}}$} &
\multicolumn{1}{c|}{AR(2)} &
\multicolumn{1}{c|}{$(\lambda_{unit},\lambda_{time},\lambda_{nn})$} &
\multicolumn{6}{c|}{bias} \\
\cline{7-12}
\multicolumn{6}{|c|}{} &
\multicolumn{1}{c|}{TROP} & \multicolumn{1}{c|}{SDID} & \multicolumn{1}{c|}{SC} &
\multicolumn{1}{c|}{DID}  & \multicolumn{1}{c|}{MC}   & \multicolumn{1}{c|}{DIFP} \\
\hline
Minimum wage & 0.770 & 0.399 & 0.604 & $(0.039,0.011)$  & $(1.6,0.35,0.011)$ & 0.113 & 0.156 & 0.165 & 0.356 & 0.187 & 0.159 \\
\hline
\end{tabular}
}
\caption{Simulation design as in Table \ref{tab:merged_panels_Tpost_Ntr} (top-panel) with ten treatment periods and ten treated units, collecting both the bias of each estimator as well as simulation for a wider range of treatments. For PWT data the outcome variable is log-GDP and for CPS data the outcome variable is log-wage in Panel B and unemployment rate in Panel C.}
\label{tab:sim_bias_only}
\end{table}

\begin{table}[H]
\captionsetup{font=scriptsize}
\begin{centering}
\scalebox{0.6}{\begin{tabular}{c|cc|cc|cc|cc|cc|cc|cc}
 & \multicolumn{2}{c|}{CPS logwage} & \multicolumn{2}{c|}{CPS urate} &  \multicolumn{2}{c|}{PWT} & \multicolumn{2}{c|}{Germany} & \multicolumn{2}{c|}{Basque} & \multicolumn{2}{c|}{Smoking} & \multicolumn{2}{c}{Boatlift} \\
\hline
 & RMSE & bias & RMSE &bias & RMSE & bias & RMSE & bias & RMSE & bias & RMSE & bias & RMSE & bias \\
\hline\hline
TROP & 0.025 & 0.014 & 0.203 & 0.104& 0.023 & 0.011 & 0.025 & 0.0003 & 0.041 & $-0.002$ & 0.085 & 0.002 & 0.115 &$-0.006$ \\
\hline
$(\lambda_{nn}=\infty)$ & 0.025 & 0.011 & 0.246 &0.197& 0.038 & 0.013 & 0.037 & 0.0005 & 0.076 & $-0.002$ & 0.124 &0.002 & 0.178 & $-0.004$\\
\hline
$(\lambda_{unit}=0)$ &  0.025 & 0.013 &  0.224 &0.182 & 0.024 & 0.013 & 0.025 & 0.0003 & 0.041 & $-0.002$ & 0.088 &0.002 & 0.115 & $-0.006$ \\
\hline
$(\lambda_{time}=0)$ & 0.032 & 0.012 & 0.226& 0.184& 0.044 & 0.034 & 0.032 & 0.0003 & 0.062 & $-0.003$ & 0.102 & 0.004 & 0.121 & $-0.007$ \\
\hline
$(\lambda_{unit}=\lambda_{time}=0)$ & 0.032 & 0.012 & 0.223 & 0.181& 0.044 & 0.034 & 0.032 & 0.0002 & 0.062 & $-0.003$ &0.102 & 0.004 &0.121 & $-0.007$ \\
\hline
$(\lambda_{unit}=0,\lambda_{nn}=\infty)$ & 0.025 & 0.008 & 0.248&0.196 & 0.039 & 0.016 & 0.037 & 0.0006 & 0.076 & $-0.002$ &0.124 & 0.002 &0.178 & $-0.004$ \\
\hline
$(\lambda_{time}=0,\lambda_{nn}=\infty)$ & 0.048 & 0.030& 0.362&0.331 & 0.123 & 0.137 & 0.071 & $-0.001$ &  0.145 & $-0.006$ & 0.192&0.009 & 0.156 &$-0.07$ \\
\hline
\makecell{DID \\ $(\lambda_{unit}=\lambda_{time}=0,\lambda_{nn}=\infty)$} & 0.049 & 0.022 & 0.384& 0.354& 0.198 & 0.175 & 0.073 & $-0.001$ & 0.145 & $-0.006$ & 0.192 & 0.009 & 0.156& $-0.007$ 
\end{tabular}
} 
\par\end{centering}
\caption{Simulation designs as in Table \ref{tab:first}, where we report the TROP estimator (first row) and the TROP estimator constraining one or combinations of $\lambda_{time}, \lambda_{unit}$ to be zero (corresponding to no time or unit weights) and $\lambda_{nn}$ to be infinity (corresponding to no regression adjustment). Each column reports the RMSE and bias across different designs in absolute terms.
 The number of treated units is 10 and the number of treated periods is 10. The treatment simulated as in Equation  \eqref{eqn:treatment_assignment} is for CPS log-wage and CPS unemployment rate is minimum wage, for PWT is democracy, and for Germany, Basque, Smoking and Boatlift is randomly assigned. } \label{tab:RMSE_and_bias_shutdown}
\end{table}

\begin{table}[H]
\captionsetup{font=scriptsize}
\begin{centering}
\scalebox{0.6}{\begin{tabular}{c|cc|cc|cc|cc|cc|cc|cc}
 & \multicolumn{2}{c|}{CPS logwage} & \multicolumn{2}{c|}{CPS urate} & \multicolumn{2}{c|}{PWT} & \multicolumn{2}{c|}{Germany} & \multicolumn{2}{c|}{Basque} & \multicolumn{2}{c|}{Smoking} & \multicolumn{2}{c}{Boatlift} \\
\hline
 & RMSE & bias &RMSE & bias & RMSE & bias & RMSE & bias & RMSE & bias & RMSE & bias & RMSE & bias \\
\hline\hline
TROP & 0.111 & 0.020 & 0.637 & $-0.022$ & 0.033 & 0.001 & 0.037 & 0.0003 & 0.035 & 0.016 & 0.178 & 0.003 & 0.350 &$-0.040$ \\
\hline
$(\lambda_{nn}=\infty)$ & 0.111 & 0.020 & 0.632 & $-0.027$ & 0.033 & 0.001 & 0.037 & 0.0003 & 0.035 & 0.016 & 0.182 &0.007 & 0.397 & $-0.037$\\
\hline
$(\lambda_{unit}=0)$ & 0.111 & 0.020 & 0.655 &$-0.099$ & 0.034 & 0.002 & 0.037 & 0.0004 & 0.038 & $-0.019$ & 0.192 &$-0.010$ & 0.391 & $-0.048$ \\
\hline
$(\lambda_{time}=0)$ & 0.112 & 0.010 & 0.632 & $-0.024$ & 0.038 & $-0.0001$ & 0.056 & $0.010$ & 0.141 & 0.121 & 0.205 & 0.029 & 0.352 & $-0.122$ \\
\hline
$(\lambda_{unit}=\lambda_{time}=0)$ & 0.115 & 0.008 &0.696 &0.023 & 0.038 & 0.0003 & 0.056 & $0.009$ & 0.046 & 0.001 &0.208 & 0.032 &0.351 & $-0.128$ \\
\hline
$(\lambda_{unit}=0,\lambda_{nn}=\infty)$ & 0.112 & 0.024 &0.630 & $-0.055$ & 0.034 & 0.002 & 0.037 & 0.0005 & 0.038 & $-0.019$ &0.193 & $-0.007$ &0.400 & $-0.046$ \\
\hline
$(\lambda_{time}=0,\lambda_{nn}=\infty)$ & 0.137 & $0.017$ &0.632 &-0.027 & 0.325 & 0.014 & 0.114 & $0.008$ & 0.228 & 0.215 & 0.601&0.212 &0.497 &$-0.338$ \\
\hline
\makecell{DID \\ $(\lambda_{unit}=\lambda_{time}=0,\lambda_{nn}=\infty)$} & 0.156 & $0.033$ &0.630 &$-0.055$ & 0.365 & 0.022 & 0.136 & $0.013$ & 0.153 & 0.044 & 0.642 & 0.202 & 0.484& $-0.320$ 
\end{tabular}
} 
\par\end{centering}
\caption{Simulation designs as in Table \ref{tab:first}, where we report the TROP estimator (first row) and the TROP estimator constraining one or combinations of $\lambda_{time}, \lambda_{unit}$ to be zero (corresponding to no time or unit weights) and $\lambda_{nn}$ to be infinity (corresponding to no regression adjustment). Each column reports the RMSE and bias across different designs in absolute terms. The number of treated units is 1 and the number of treated periods is 1. The treatment is simulated as in Equation  \eqref{eqn:treatment_assignment} is for CPS log-wage and CPS unemployment rate is minimum wage, for PWT is democracy, and for Germany, Basque, Smoking and Boatlift is randomly assigned.} \label{tab:one_shutdown}
\end{table}

\end{document}

%% file: guido_macros.tex
\newtheorem{result}{Result}
\newtheorem{example}{Example}
\newtheorem{comments}{Comment}

% Argmax and Argmin

\iffalse
\DeclareMathOperator*{\argmin}{arg\,min}
\DeclareMathOperator*{\argmax}{arg\,max}

\setcounter{MaxMatrixCols}{10}
\renewcommand{\mathbf}{\boldsymbol}
\setlength{\topmargin}{-0.25in}
\setlength{\textheight}{8.75in}
\setlength{\evensidemargin}{-0.125in}
\setlength{\oddsidemargin}{-0.125in}
\setlength{\textwidth}{6.75in}
\renewcommand{\thepage}{}
\renewcommand{\appendix}{\footnotesize\parindent 0cm\setcounter{equation}{0}
	\renewcommand{\theequation}{A.\arabic{equation}}
	\setcounter{lemma}{0}\renewcommand{\thelemma}{A.\arabic{lemma}}}

\fi

%\numberwithin{result}{chapter}
%\numberwithin{proposition}{chapter}
%\numberwithin{theorem}{chapter}
%\numberwithin{assumption}{chapter}
%\numberwithin{lemma}{chapter}
%\numberwithin{definition}{chapter}
%\numberwithin{example}{chapter}

\newcommand{\ar}{{\rm AR}}
\newcommand{\ad}{{\rm ad}}
\newcommand{\ability}{{\rm ability}}
\newcommand{\always}{{\rm always-taker}}
\newcommand{\aand}{{\hskip1cm {\rm and}\ \ }}
\newcommand{\ave}{{\rm ave}}
\newcommand{\aven}{\frac{1}{N}\sum_{i=1}^N}
\newcommand{\alphas}{\alpha^*}
\newcommand{\atsign}{$0$}

\newcommand{\bootstrap}{\mathrm{boot}}
\newcommand{\bepsilon}{\mathbf{\varepsilon}}
\newcommand{\been}{{\bf 1}}
\newcommand{\betas}{\beta^*}
\newcommand{\betablp}{\beta_\blp}
\newcommand{\bg}{{\bf G}}
\newcommand{\bi}{{\bf I}}
\newcommand{\blp}{{\rm blp}}
\newcommand{\bp}{{\bf P}}
\newcommand{\bpx}{{{\bf P}_{\bf X}}}
\newcommand{\by}{{\bf Y}}
\newcommand{\bx}{{\bf X}}
\newcommand{\br}{{\bf R}}
\newcommand{\bw}{{\bf W}}
\newcommand{\bww}{{\bf w}}
\newcommand{\brr}{{\bf r}}
\newcommand{\ba}{{\bf A}}
\newcommand{\bz}{{\bf Z}}
\newcommand{\bv}{{\bf V}}
\newcommand{\bs}{{\bf S}}

\newcommand{\cals}{{\cal S}}
\newcommand{\calh}{{\cal H}}
\newcommand{\calt}{{\cal T}}
\newcommand{\ccc}{{\cal C}}
\newcommand{\ccb}{{\cal S}}
\newcommand{\ccr}{{\cal R}}
\newcommand{\caln}{{\cal N}}
\newcommand{\calg}{{\cal G}}
\newcommand{\calp}{{\cal P}}
\newcommand{\ccii}{{\rm CI}}
\newcommand{\ca}{\mathbb{A}}
\newcommand{\ci}{{\rm CI}}
\newcommand{\comp}{{\rm complier}}
\newcommand{\ct}{{\rm c}}
\newcommand{\cluster}{{\rm cluster}}
\newcommand{\combined}{{\rm comb}}

\newcommand{\data}{\mathrm{data}}
\newcommand{\define}{:=}
\newcommand{\did}{\mathrm{did}}
\newcommand{\donald}{\mathrm{Donald}}
\newcommand{\dif}{{\rm dif}}

\newcommand{\emphunderline}{\underline}
\newcommand{\earn}{{\rm earnings}}
\newcommand{\educ}{{\rm educ}}
\newcommand{\exper}{{\rm exper}}
\newcommand{\expers}{{\rm exper}^2}
\newcommand{\el}{{\rm el}}
\newcommand{\ehw}{{\rm EHW}}
\newcommand{\ehww}{{\rm ehw}}

\newcommand{\ff}{{\rm f}}
\newcommand{\frd}{{\rm frd}}
\newcommand{\finsin}{\frac{1}{N}\sum_{i=1}^N}
\newcommand{\fin}{\frac{1}{N}}

\newcommand{\gary}{\mathrm{Gary}}
\newcommand{\gls}{{\rm fgls}}
\newcommand{\gs}{g}
\newcommand{\GG}{G}

\newcommand{\hmmv}{{\hat{\mathbb{V}}}}
\newcommand{\high}{{\rm high}}
\newcommand{\hct}{{\rm HC2}}
\newcommand{\hcth}{{\rm HC3}}
\newcommand{\homo}{\mathrm{homo}}
\newcommand{\hatbetaols}{\hat\beta^{\rm ols}}
\newcommand{\hatthetagmm}{\hat\theta_{\rm gmm}}
\newcommand{\hatthetael}{\hat\theta_{\rm el}}
\newcommand{\hatthetaml}{\hat\theta_{\rm ml}}
\newcommand{\htau}{\hat{\tau}}

\newcommand{\iv}{\mathrm{IV}}
\newcommand{\indep}{\perp\!\!\!\perp}

\newcommand{\josh}{\mathrm{Josh}}
\newcommand{\jack}{\mathrm{jacknife}}

\newcommand{\lambdas}{\lambda^*}
\newcommand{\lb}{{\rm lb}}
\newcommand{\liml}{{\rm liml}}
\newcommand{\learn}{{\rm log(earnings)}}
\newcommand{\lectures}{{Lectures in Econometrics,\ }}

\newcommand{\mrelec}{\mathrm{elec}}
\newcommand{\mrcap}{\mathrm{cap}}
\newcommand{\mroper}{\mathrm{oper}}
\newcommand{\mrno}{\mathrm{no}}
\newcommand{\mrgas}{\mathrm{gas}}
\newcommand{\medd}{{\rm med}}
\newcommand{\med}{{\rm med}}
\newcommand{\mle}{{\rm mle}}
\newcommand{\mmr}{{\mathbb{R}}}
\newcommand{\mmw}{{\mathbb{W}}}
\newcommand{\mmz}{\mathbb{Z}}
\newcommand{\mx}{\mathbb{X}}
\newcommand{\mme}{{\mathbb{E}}}
\newcommand{\mmv}{{\mathbb{V}}}
\newcommand{\mma}{{\mathbb{A}}}
\newcommand{\mmva}{{\mathbb{AV}}}
\newcommand{\mmc}{{\mathbb{C}}}
\newcommand{\mmx}{{\mathbb{X}}}
\newcommand{\modrobust}{{\rm HC2}}
\newcommand{\mvar}{\mathbb{V}}
\newcommand{\mmav}{{\mathbb{AV}}}

\newcommand{\nfn}{N_{{\rm f}\ct}}
\newcommand{\nmn}{N_{{\rm m}\ct}}
\newcommand{\nfe}{N_{{\rm f}\tc}}
\newcommand{\nme}{N_{{\rm m}\tc}}
\newcommand{\never}{{\rm never-taker}}
\newcommand{\nc}{N_{\ct}}
\newcommand{\nt}{N_{\tc}}
\newcommand{\nf}{N_{\ff}}
\newcommand{\neyman}{\mathrm{neyman}}

\newcommand{\opsn}{o_p\left(N^{-1/2}\right)}
\newcommand{\Opsn}{O_p\left(N^{-1/2}\right)}
\newcommand{\opn}{o_p\left(N^{-1}\right)}
\newcommand{\Opn}{O_p\left(N^{-1}\right)}
\newcommand{\obs}{{\rm obs}}
\newcommand{\oy}{\overline{Y}}
\newcommand{\orr}{\overline{R}}
\newcommand{\ols}{{\rm ols}}

\newcommand{\pop}{{\rm pop}}
\newcommand{\pate}{{\rm pate}}
\newcommand{\pos}{{\rm pos}}
\newcommand{\popt}{{\rm patt}}
\newcommand{\pr}{{\rm pr}}

\newcommand{\rank}{{\rm rank}}
\newcommand{\reg}{{\rm reg}}
\newcommand{\robust}{{\rm robust}}
\newcommand{\ress}{Y_i-X_i hatbetaols}
\newcommand{\rmelec}{{\rm elec}}
\newcommand{\relec}{{\rm elec}}
\newcommand{\rmoper}{{\rm oper}}
\newcommand{\rmcap}{\mathrm{cap}}
\newcommand{\rmgas}{\mathrm{gas}}
\newcommand{\rmno}{\mathrm{no}}
\newcommand{\rmif}{{\rm if}}

\newcommand{\qob}{{\rm qob}}

\newcommand{\subs}{\mathrm{subsampling}}
\newcommand{\sate}{{\rm sate}}
\newcommand{\sample}{{\rm sample}}
\newcommand{\samplet}{{\rm satt}}
\newcommand{\spp}{{\rm sp}}
\newcommand{\srd}{{\rm srd}}
\newcommand{\se}{{\rm s.e.}}
\newcommand{\strata}{{\rm strat}}
\newcommand{\snn}{\sum_{i=1}^N}
\newcommand{\snt}{\sum_{t=1}^T}
\newcommand{\str}{^{*}}

\newcommand{\tc}{{\rm t}}
\newcommand{\tsls}{{\rm tsls}}
\newcommand{\tmle}{\hat\theta^{\rm mle}}
\newcommand{\thetaml}{\theta^{\rm mle}}
\newcommand{\ty}{\tilde Y}
\newcommand{\tick}{\checkmark}
\newcommand{\thetas}{\theta^*}
\newcommand{\taup}{\tau_\spp}
\newcommand{\tautp}{\tau_{\spp,\tc}}

\newcommand{\ow}{\overline{W}}
\newcommand{\oz}{\overline{Z}}
\newcommand{\ox}{{\overline{X}}}

\newcommand{\ub}{{\rm ub}}

\newcommand{\veen}{{\rm homo}}
\newcommand{\vtweea}{{\rm    homo,unbiased}} 
\newcommand{\vtwee}{{\rm homo,unbiased,t-dist}}
\newcommand{\vdrie}{{\rm robust}} 
\newcommand{\vet}{{\rm veteran}}

\newcommand{\wb}{{\bf W}}
\newcommand{\win}{W_{i}}
\newcommand{\welch}{\mathrm{welch}}

\newcommand{\xbj}{{\bf X}_{(j)}}
\newcommand{\xb}{{\bf X}}
\newcommand{\xbi}{{\bf X}_{(i)}}
\newcommand{\xxb}{{\bf x}}

\newcommand{\yoin}{Y_{i}}
\newcommand{\yin}{Y_i(0)}
\newcommand{\yie}{Y_i(1)}
\newcommand{\yinn}{Y_{i}(0)}
\newcommand{\yien}{Y_{i}(1)}
\newcommand{\ybn}{{\bf Y}(0)}
\newcommand{\ybe}{{\bf Y}(1)}
\newcommand{\yb}{{\bf Y}}
\newcommand{\ybo}{{\bf Y}^{\rm obs}}
\newcommand{\ybm}{{\bf Y}^{\rm mis}}
\newcommand{\yybn}{{\bf y}(0)}
\newcommand{\yybe}{{\bf y}(1)}
\newcommand{\ybni}{{\bf Y}_{(i)}(0)}
\newcommand{\ybei}{{\bf Y}_{(i)}(1)}
\newcommand{\ybnj}{{\bf Y}_{(j)}(0)}
\newcommand{\ybej}{{\bf Y}_{(j)}(1)}
\newcommand{\yoi}{Y^{\rm obs}_i}

\newcommand{\zin}{Z_{i}}